\documentclass[11pt]{article}
\usepackage[hmargin=1in,vmargin=1in]{geometry}
\usepackage{hchang}
\usepackage{thm-restate} 
\usepackage{cleveref}
\usepackage{comment}
\usepackage{tcolorbox}
\usepackage{soul} 

\nolinenumbers
\usepackage{makecell}

\usepackage{amsopn}

\newtheorem{theorem}{Theorem}[section]
\newtheorem{lemma}[theorem]{Lemma}
\newtheorem{claim}[theorem]{Claim}
\newtheorem{observation}[theorem]{Observation}

\newtheorem{remark}[theorem]{Remark}

\newtheorem{definition}[theorem]{Definition}
\newtheorem{question}[theorem]{Question}

\newcommand{\eps}{\varepsilon}
\newcommand\mst{\mathrm{MST}}

\definecolor{BrickRed}{rgb}{0.8, 0.25, 0.33}
\def\EMPH#1{\emph{\textcolor{BrickRed}{#1}}}
\newcommand{\lbl}[1]{\mathfrak{#1}}

\title{Optimal Bounds for Spanners and Tree Covers in Doubling Metrics}

\date{}
\author{%
An La%
\thanks{Manning CICS, UMass Amherst. Email: {\tt anla@umass.edu}}
\and 
Hung Le%
\thanks{Manning CICS, UMass Amherst. Email: {\tt hungle@cs.umass.edu}}  
\and
Shay Solomon%
\thanks{Tel Aviv University. Email: {\tt shayso@tauex.tau.ac.il}}  
\and
Cuong Than%
\thanks{Manning CICS, UMass Amherst. Email: {\tt cthan@cs.umass.edu}}  
\and
Vinayak%
\thanks{Manning CICS, UMass Amherst. Email: {\tt vvinayak@umass.edu}}
\and
Shuang Yang%
\thanks{Manning CICS, UMass Amherst. Email: {\tt shuangyang@umass.edu}}
\and
Tianyi Zhang%
\thanks{State Key Laboratory for Novel Software Technology, Nanjing University, Nanjing, China. Email: {\tt tianyi.zhang@inf.ethz.ch}}
}

\begin{document}

\maketitle

\begin{abstract} 
    It is known that any $n$-point set in the $d$-dimensional Euclidean space $\mathbb{R}^d$, for $d = O(1)$, admits:
    \begin{enumerate}
        \item A $(1+\eps)$-spanner with 
    maximum degree
$\tilde{O}(\eps^{-d+1})$ and with lightness $\tilde{O}(\eps^{-d})$, for any  $\eps > 0$.\footnote{The {\em lightness} is a normalized notion of weight, where we divide the spanner weight by the weight of a minimum spanning tree. Here and throughout, the $\tilde{O}$ and $\tilde{\Omega}$ notations hide $\texttt{polylog}(\eps^{-1})$ terms.} 
\item A $(1+\eps)$-tree cover with $\tilde{O}(n \cdot \eps^{-d+1})$ trees and maximum degree of $O(1)$ in each tree.
\end{enumerate}
Moreover, all the parameters in these constructions are optimal:
For any $2 \le d = O(1)$,
there exists an $n$-point set in $\mathbb{R}^d$,  for which any $(1+\eps)$-spanner has $\tilde{\Omega}(n \cdot \eps^{-d+1})$ edges and lightness $\tilde{\Omega}(\eps^{-d})$.
The upper bounds for Euclidean spanners rely heavily on the spatial property of {\em cone partitioning} in $\mathbb{R}^d$, which does not seem to extend to the wider family of {\em doubling metrics}, i.e., metric spaces of constant {\em doubling dimension}. 
In doubling metrics, a  {\bf simple spanner construction from two decades ago, 
the {\em net-tree spanner}}, has $\tilde{O}(n \cdot \eps^{-d})$ edges, and it could be transformed into a spanner of maximum degree $\tilde{O}(\eps^{-d})$ and 
lightness 
$\tilde{O}(n \cdot \eps^{-(d+1)})$ by pruning redundant edges. Moreover, a careful refinement of the net-tree spanner yields a $(1+\eps)$-tree cover with $\tilde{O}(\eps^{-d})$ trees.
Despite a large body of work,
the problem of obtaining 
tight bounds for spanners and tree covers
 in the wider family of {\em doubling metrics} has remained elusive.
We resolve this problem by presenting:
\begin{enumerate}
    \item A surprisingly simple and tight lower bound, which shows  that the net-tree spanner and its pruned version are optimal with respect to all the involved parameters.
    \item A new construction of $(1+\eps)$-tree covers with $\tilde{O}(n \cdot \eps^{-d})$ trees, with maximum degree  $O(1)$ in each tree. This construction is optimal with respect to the number of trees and  maximum degree.  
\end{enumerate}

\end{abstract} 

\section{Introduction}
\label{sec:intro}

\subsection{Low-Dimensional Euclidean Spaces}
\subsubsection{Euclidean Spanners}
Let $P$ be a set of $n$ points in the Euclidean space $\mathbb R^d, d \ge 2$, and consider the complete weighted graph $G_P = (P,{P \choose 2},\|\cdot\|)$ induced by $P$,
where the weight of any edge $(x,y) \in {P \choose 2}$ is  the Euclidean distance $\|xy\|$ between its endpoints.
We say that a spanning subgraph  $H = (P,E,\|\cdot\|)$ of $G_P$ (with $E \subseteq {P \choose 2}$) is a \emph{$t$-spanner} for  $P$, for a parameter $t \ge 1$ that is called the {\em stretch} of the spanner, if 
$d_H(x,y) \le t \cdot \|xy\|$ holds
$\forall x,y \in P$.
Spanners for Euclidean spaces, or {\em Euclidean spanners},
were introduced
in the pioneering work of Chew~\cite{Chew86} from 1986, which gave an $O(1)$-spanner with $O(n)$ edges.
The first constructions of Euclidean $(1+\eps)$-spanners, for any parameter $\eps > 0$, were given in the seminal works of \cite{Clarkson87,Keil88,KG92}) that introduced the {\em $\Theta$-graph} in 2 and 3-dimensional Euclidean spaces,
which was   generalized for any Euclidean space $\mathbb R^d$ in \cite{RS91,ADDJS93}.
The $\Theta$-graph is a natural variant of the {\em Yao graph},  introduced  by Yao~\cite{Yao82} in 1982, and can be described as follows.

\begin{tcolorbox}[colback=gray!10, colframe=gray!40]
{\bf Yao graph:} $\forall p \in P$, the space $\mathbb R^d$ around $p$ is partitioned into {\em cones} of angle $\Theta$ each ($O(\Theta^{-d+1}$) cones for $d = O(1)$),
and then edges are added between point $p$ and its closest point in each of these cones.  
\end{tcolorbox}

The $\Theta$-graph is defined similarly to the Yao graph: instead of connecting $p$ to its closest point in each cone, connect it to a point
whose {\em orthogonal projection} to some fixed ray contained in the cone is closest to $p$. 
Taking $\Theta$ to be $c  \eps$, for small enough constant $c$, one can show that the stretch of the $\Theta$ and Yao graphs is at most $1+\eps$.
Since the number of cones is asymptotically $\eps^{-d+1}$ (for $d = O(1)$), the {\em maximum degree} is $O(\eps^{-d + 1})$, leading to a total of $O(n \cdot \eps^{-d+1})$ edges.

The tradeoff  between stretch $1+\eps$ and $O(n \cdot \eps^{-d+1})$ edges (and maximum degree $O(\eps^{-d + 1})$) is also achieved by other constructions,
including the {\em greedy spanner}~\cite{ADDJS93,CDNS92,NS07} and the
gap-greedy spanner~\cite{Salowe92,AS97}. 
The spatial cone partitioning of $\mathbb{R}^d$ is key to attaining the size bound  $O(n \cdot \eps^{-d+1})$ in  these constructions, either in the constructions themselves or in their analysis. In 2019, Le and Solomon \cite{LS19} showed that this stretch-size tradeoff is {\em existentially tight}: For any constant $d \geq 2$, there exists an $n$-point set in $\mathbb{R}^d$ (basically a set of evenly spaced points on the $d$-dimensional sphere), for which any $(1+\eps)$-spanner has $\Omega(n \cdot \eps^{-d+1})$ edges.

{\bf Light spanners.~} Another basic property of spanners is {\em lightness},
defined as the ratio of a spanner's total weight to the weight of the Minimum Spanning Tree of $P$. 
A long line of work 
\cite{ADDJS93,DHN93,DNS95,RS98,NS07,BLW19,LS19}, starting from the paper of Das et al.~\cite{DHN93} in 1993,
showed that for any point set in $\mathbb{R}^d$, the greedy $(1+\eps)$-spanner
of \cite{ADDJS93}
has constant (depending on $\eps$ and $d$) lightness. 
The exact dependencies on $\eps$ and $d$ in the lightness bound were not explicated in \cite{ADDJS93,DHN93,DNS95}.
In their seminal work on approximating TSP in $\mathbb{R}^d$ using light spanners, Rao and Smith~\cite{RS98} showed that the greedy spanner has lightness $\eps^{-O(d)}$, and they raised the question of determining the exact constant hiding in the exponent of their upper bound.
The proofs in \cite{ADDJS93,DHN93,DNS95,RS98} were incomplete; the first complete proof was given in \cite{NS07}, where a lightness bound of $O(\eps^{-2d})$ was established.
This line of work culminated with the work of Le and Solomon \cite{LS19}, which
improved the lightness bound to
$\tilde{O}(\eps^{-d})$, where we shall use the $\tilde{O}$ and $\tilde{\Omega}$ notations to suppress  \texttt{polylog}$(\eps^{-1})$ terms.  The exact lightness bound here is
$O(\eps^{-d} \cdot \log(1/\eps))$, but for the sake of brevity we will mostly disregard \texttt{polylog}$(\eps^{-1})$ terms from now on.
They
  \cite{LS19} also showed that this stretch-lightness tradeoff is {\em existentially tight} (up to a $\log(\eps^{-1})$ factor): for any constant $d \geq 2$, there exists an $n$-point set in $\mathbb{R}^d$ (the same set of evenly spaced points on the $d$-dimensional sphere), for which any $(1+\eps)$-spanner has lightness $\Omega(\eps^{-d})$.

\subsubsection{Euclidean Tree Covers}

A more structured variant of spanners is the {\em tree cover}, defined as a collection $\mathcal{T}$ of trees for a given graph or metric space such that, for every pair of vertices, there exists a tree in $\mathcal{T}$ that preserves their distance up to a given stretch factor, without shortening distances.
Due to their strong structural properties, tree covers play a pivotal role in routing, (path-reporting) distance oracles, and various other algorithmic applications. 
The first tree cover construction for a graph or metric family was given by Arya et al.~\cite{ADMSS95}---known as the ``Dumbbell Theorem''---asserting that any low-dimensional Euclidean space admits a $(1 + \eps)$-stretch tree cover of size $O(\eps^{-d} \cdot \log(1/\eps))$, for any $\eps > 0$. 
The lower bound by Le and Solomon~\cite{LS22} on the size of spanners directly implies that any $(1+\eps)$-stretch tree cover must have size at least $\Omega(\eps^{-d + 1})$. 
Recently, Chang et al.~\cite{CCL+24b} closed this longstanding gap between the upper and lower bounds on the size of tree covers, up to a $\log(1/\eps)$ factor, by improving the size upper bound to $O(\eps^{-d + 1} \cdot \log(1/\eps))$. 
Another important quality measure of tree covers is the maximum degree of a vertex in any tree; in particular, a constant degree tree cover enables compact routing schemes with small routing tables~\cite{CCL+24b}. 
Chang et al.\ constructed a constant degree tree cover without increasing the size beyond $O(\eps^{-d + 1}\cdot \log(1/\eps))$. 
Whether the $\log(1/\eps)$ factor separating this upper bound from the lower bound is necessary remains an open problem.

\subsection{Doubling Metrics}

\subsubsection{Doubling Spanners}
Euclidean spanners have been extensively studied over the years \cite{Keil88,KG92,ADDJS93,DN94,ADMSS95,DNS95,AS97,RS98,GLN02,
GGN04,AWY05,CG06,GR082,CGMZ16,DES08,Sol14,ES15,LS19},
with a plethora of applications, such as in geometric approximation algorithms \cite{RS98,GLNS02,GLNS02b,GLNS08}, geometric distance oracles
\cite{GLNS02,GLNS02b,GNS05,GLNS08}, network design \cite{HP00,MP00} and machine learning \cite{GKK17}.
(See the book \cite{NS07} for an excellent account on  Euclidean spanners and their applications.)
There is a growing body of work on {\em doubling spanners}, i.e., spanners for the wider family of  {\em doubling metrics};\footnote{The {\em doubling dimension} of a metric space is the smallest value $d$ such that every ball $B$ in the metric can be covered by at most $2^d$ balls of half the radius of $B$; a metric space is called {\em doubling} if its doubling dimension is constant.
The doubling dimension generalizes the standard Euclidean dimension, as the doubling dimension of the Euclidean space $\mathbb{R}^d$ is $\Theta(d)$.} 
see \cite{GGN04, CGMZ05, CG06, HPM06, GR081, GR082, Smid09, CLN12a, ES15, CLNS13, Sol14, BLW19, LT22, KLMS22, LST23}, and the references therein. {\bf A common theme in this line of work is to  devise constructions of spanners for doubling metrics that are just as good as the analog Euclidean spanner constructions}. Alas, this may not always be possible, as doubling metrics do not possess the spatial properties of Euclidean spaces---and in particular the spatial property of {\em cone partitioning}, which is key to achieving the aforementioned stretch-size and stretch-lightness upper bounds. Despite this shortcoming, the basic {\em packing bound}   in doubling metrics can be used to construct, via a simple greedy procedure, a {\em hierarchy of nets}, which induces the so-called {\em net-tree} \cite{HPM06, GGN04, CGMZ05}.
(See \Cref{sec:prel} for the packing bound and other definitions.)
Equipped with such a hierarchy of nets, a $(1+\eps)$-spanner with $\tilde{O}(n \cdot \eps^{-d})$ edges is constructed as follows \cite{HPM06, GGN04, CGMZ05}. 

\begin{tcolorbox}[colback=gray!10, colframe=gray!40] {\bf Net-tree spanner~\cite{HPM06, GGN04, CGMZ05}:} Let $X = N_0\supseteq N_1\supseteq N_2 \supseteq\ldots \supseteq N_{\ceil{\log\Delta}}$ be a hierarchy of nets of a doubling metric $(X,d_X)$ with minimum distance $1$ and maximum distance $\Delta$, where $N_i$ is a $2^{i}$-net for $N_{i-1}$, $1\leq i \leq \ceil{\log\Delta}$. \\Let $E_{-1} = \emptyset$ and for each $0\leq i \leq \ceil{\log\Delta}$: 
\begin{equation}
    E_i = \left\{(x,y)~|~ x,y\in N_i, d_X(x,y) \leq \left(4+ \frac{32}{\eps}\right)2^i\right\}\setminus E_{i-1}
\end{equation}
Then $(X,\cup_{0\leq i\leq \ceil{\log\Delta}}E_i,d_X)$ is a $(1+\eps)$-spanner of $(X,d_X)$.
\end{tcolorbox}

{\bf The net-tree spanner is a simple and basic spanner construction from over 20 years ago}, and it provides the state-of-the-art size bound, $O(n \cdot \eps^{-d} \cdot \log(1/\eps))$. 
Le and Solomon \cite{LS23} presented a unified framework for transforming sparse spanners into light spanners by carefully pruning redundant edges. In particular, for doubling metrics, the framework of \cite{LS23} provides a pruned spanner, with a lightness bound 
of $\tilde{O}(\eps^{-(d+1)})$, which is the state-of-the-art lightness bound for doubling metrics.
Interestingly, these  upper bounds on the size and lightness in doubling metrics exceed the respective Euclidean bounds  by a factor of $\eps^{-1}$. 
Despite a large body of work in the area, 
it has remained a longstanding open question whether the superior Euclidean bounds of $\tilde{O}(n \cdot \eps^{-d+1})$ edges and lightness $\tilde{O}(\eps^{-d})$
could be achieved also in doubling metrics.

\begin{question}
\label{ques:eucli-doub-match}
Can one get a construction of $(1+\eps)$-spanners in doubling metrics with $\tilde{O}(n \cdot \eps^{-d+1})$ edges and/or lightness $\tilde{O}(\eps^{-d})$? 
\end{question}

This question is related to a possibly deeper question, regarding the (im)possibility to generalize the spatial cone partitioning in $\mathbb{R}^d$ to arbitrary doubling metrics. 

\subsubsection{Doubling Tree Covers}

Bartal et al.~\cite{BFN19} constructed a tree cover of size $O(\varepsilon^{-d} \cdot \log(1/\varepsilon))$ using the net-tree by redistributing cross edges at each level into color classes. However, in their construction, a vertex in one of the trees may have many children, resulting in unbounded degree, potentially as large as $\log(\Delta)$, where $\Delta$ is the aspect ratio. Motivated by the Euclidean setting, it is natural to ask whether one can obtain a tree cover with bounded degree in doubling metrics.

\begin{question}
\label{ques:bound-deg-cover}
Can one construct a $(1+\varepsilon)$-tree cover in doubling metrics whose maximum degree is an absolute constant, i.e., $O(1)$, while having size $\tilde{O}(\varepsilon^{-d})$?
\end{question}

\subsection{Our contribution}

We give a negative answer to \Cref{ques:eucli-doub-match} by a {\bf surprisingly simple} and tight lower bound.

\begin{restatable}{theorem}{LowerBoundMain}
\label{thm:main}
For any integer constant $d \ge 1$, parameter $\varepsilon \in (0,1)$, and $n \in \mathbb{Z}^+$ satisfying $\varepsilon^{-d} = O(n)$, there exists an $n$-point ultrametric space $(X, d_X)$ of doubling dimension $d$ such that any $(1+\varepsilon)$-spanner $G = (X, E, w)$ for $X$ has $\Omega(n \cdot \varepsilon^{-d})$ edges and lightness $\tilde{\Omega}(\varepsilon^{-(d+1)})$.
\end{restatable}
Since the net-tree spanner has sparsity \footnote{The \emph{sparsity} of a spanner $H$ is the ratio between number of edges in $H$ over $n - 1$, for $n > 1$.} $O(\varepsilon^{-d})$~\cite{GGN04}, our lower bound shows, in particular, that \textbf{the net-tree spanner is optimal}; this is our {\bf key conceptual contribution}.

 Second, {\bf as our main technical contribution}, we show that a tree cover with constant degree and nearly optimal size is achievable, resolving \Cref{ques:bound-deg-cover} affirmatively.

\begin{restatable}{theorem}{TreeCoverBoundedDegree}
\label{thm:tree-cover}
For any metric $(X, d_X)$ with doubling dimension $d$, parameter $0 < \varepsilon < 1$, there exists a $(1+\varepsilon)$-stretch tree cover $\mathcal{T}$ of size $\tilde{O}\!\left(\varepsilon^{-d}\right)$,
where each tree in $\mathcal{T}$ has maximum vertex degree $O(1)$.
\end{restatable}

To complement our result, we show (in~\Cref{sec:lowerbound_maxdeg}) that it is impossible to eliminate the $\mathsf{polylog}(1/\varepsilon)$ factor from the size of any constant-degree tree cover.
\begin{restatable}{theorem}{LowerBoundMaxDegree}
\label{thm:bdd-deg}
For any integer constant $d \ge 1$, any parameter $\varepsilon \in (0,1)$, and any $n \in \mathbb{Z}^+$ satisfying $\varepsilon^{-d} = O(n)$, there exists an $n$-point ultrametric space $(X, d_X)$ of doubling dimension $d$ such that any $(1+\varepsilon)$-spanner $G = (X, E, w)$ for $X$ has maximum degree $\Omega(\log(1/\varepsilon) \cdot \varepsilon^{-d})$.
\end{restatable}
\Cref{thm:bdd-deg} demonstrates an inherent multiplicative gap of $\log(1/\varepsilon)$ between the maximum degree and the optimal sparsity of $(1+\eps)$-spanners in doubling metrics, a gap that does not appear in Euclidean spaces. Although \Cref{thm:bdd-deg} is not the main focus of our work, it provides an interesting separation between Euclidean and doubling metrics. Moreover, it shows that the $\mathsf{polylog}(1/\varepsilon)$ term in \Cref{thm:tree-cover} (hidden via the $\tilde{O}$-notation) is unavoidable in doubling metrics, whereas its necessity in Euclidean spaces remains an open question.

\subsection{Technical Overview}
Our tree cover construction
provides the key technical contribution of this work. In what follows, we provide a  technical overview of this construction. 

Our tree cover construction relies on a hierarchical structure---specifically, the standard {\em net-tree}---and builds a tree cover so that every cross edge $e$ (i.e., every edge in $E_i$) is contained in some tree. For every pair of points $(u,v)$, there exists a cross edge $(u',v')$ such that the path from $u$ to $v$ that goes through $(u',v')$ provides a good approximation of $d_X(u,v)$. 
The distance between $u$ and $v$ is then preserved by the tree $T$ that contains $(u',v')$; in this case, we say that $(u, v)$ is covered by $T$.

The construction proceeds in two phases. In Phase~1, we partition the cross edges into congruent classes according to their levels modulo $\log(1/\varepsilon)$, and handle each class independently. For each class, we distribute the cross edges into each tree in our tree cover. Note that after Phase~1, we already obtain a tree cover of stretch $(1 + \eps)$ and size $\tilde{O}(\eps^{-d})$. In Phase~2, for a fixed level, we replace edges in a controlled manner to obtain the required bounded degree property, without significantly increasing the stretch. The desired $(1+\varepsilon)$ stretch bound can then be recovered by a simple scaling argument.

Our Phase~2 (\Cref{sec:treecover_bounded_deg} and~\Cref{sec:treecover_bdd_deg_appendix} in detail) is inspired by the degree-reduction procedure of \cite{CCL+24b}. To achieve constant degree, we proceed in two steps. First, using the technique in \cite{CGMZ05}, we reroute the edges in each tree $T$ in the cover produced by Phase~1, obtaining trees of maximum degree $\varepsilon^{-O(d)}$, while the stretch increases only to $1+O(\varepsilon)$. Second, we use the rerouting technique in \cite{CCL+24b} to obtain absolute constant maximum degree. For each vertex $u$ in each such tree $T$, we replace the edges connecting $u$ to its children by a binary tree. Suppose that $l$ is the maximum distance from (the representative of) $u$ to its children, this step incurs an additive distortion of $\Theta(\log(\varepsilon^{-d})\cdot l)$ for every pair of points covered by $T$.

A key technical challenge in integrating Phase~2 into Phase~1 is that the edges replaced by binary trees must be sufficiently short relative to the cross edges at the same level. Otherwise, the additive distortion for a pair $(u,v)$ might become comparable to $d(u,v)$ itself, producing a stretch larger than~2. In the construction of \cite{BFN19}, a point $u$ in a tree $T$ may be incident to $\varepsilon^{-d}$ cross edges of nearly equal length; applying Phase~2 on top of \cite{BFN19} would therefore not yield a tree cover with good stretch.

Our first idea is to ensure that, in Phase~1, for every level we partition the cross edges into $O(\varepsilon^{-d})$ sets such that in each tree of the resulting cover, every point $u$ is incident to \emph{at most one} edge from each set. This motivates partitioning $E_i$ into sets of \emph{matchings}.

We implement Phase~1 as follows. We initialize tree cover $\mathcal{T}$ with $\tilde{O}(\varepsilon^{-d})$ empty forests. For each congruent class modulo $\log(1/\varepsilon)$, we process the levels in increasing order. For each level $i$, we partition $E_i$ into sets of matchings $\mathcal{M}_i$. For each matching $M \in \mathcal{M}_i$, we assign $M$ to forest $F \in \mathcal{T}$, and add all edges in $M$ to $F$. We continue this process until the top level.

A second challenge is ensuring that the graph $F$ remains a forest after inserting the edges of $M$. Specifically, we guarantee that for any two edges $(u,v)$ and $(u',v')$ in $M$, there is no tree in the forest $F$ containing both $u$ and $u'$. This leads to our second idea: each matching $M$ must be \emph{pairwise far}. That is, for every $M\in\mathcal{M}_i$ and every two distinct edges $e,e'\in M$, the distance between (any two endpoints of) $e$ and $e'$ must be at least $\Theta(2^i)$. Under this condition, an inductive argument shows that every tree in each forest $F$ has small diameter ($< 2^i$); thus, after connecting each tree to an edge in $M$, a sufficient separation still remains, ensuring that the resulting graph is indeed a forest. Such matchings, also known as {\em red–blue matchings}, were also used by \cite{FL22, LL25} in different contexts.

\section{Preliminaries~} \label{sec:prel}
Let $(X, d_X)$ be a metric space. 
The {\em aspect ratio} $\Delta$ of $X$ is the ratio of its largest to smallest pairwise distances. The following definitions are crucial for our construction and proofs.

\begin{definition}[$r$-net]
    For a metric space $(X, d_X)$, an $r$-net $N\subseteq X$ is a set satisfying:
    \begin{enumerate}
        \item \textbf{[Packing]} for every distinct $u,v \in N$, $d_X(u, v) > r$
        \item \textbf{[Covering]} for every $x\in X$, there exists $u\in N$ such that $d_X(u, x) \leq r$.
    \end{enumerate}
\end{definition}

\begin{definition}[$(\alpha, \beta)$-Net Tree]\label{def:alpha-beta-net-tree}
    An $(\alpha, \beta)$-net tree $\tau$ is induced by a hierarchy of nets $N_{-1}=X, N_0, N_1,\ldots, N_{\log\Delta}$ at levels $-1, 0, 1,\ldots, \log\Delta$ of $\tau$ resp., satisfying the following:
    \begin{enumerate}
        \item (Hierarchy) $N_i\subseteq N_{i-1}$ for every $i\in[\log\Delta]$.
        \item (Net) For every $i\in \{0,1,
        \ldots, \log\Delta\}$, $N_i$ is an $(\alpha\cdot\beta^i)$-net of $N_{i-1}$.
    \end{enumerate}
\end{definition}

The following lemma gives the basic packing bound of doubling metrics.
\begin{lemma}[Packing Lemma] 
\label{packing_lemma}\label{lm:packing}
    Metric $(X, d_X)$ has doubling dimension $d$.
    If $S \subseteq X$ is a subset with pairwise distance at least $r$ that is contained in a ball of radius $R$, then $|S| \leq \left(\frac{4R}{r}\right)^{d}$.
\end{lemma}

\begin{definition}[Ultrametric]
    An {\em ultrametric} $(X, d_X)$ is a metric space that satisfies the strong triangle inequality, i.e., for any $u, v, w\in X$, we have $d_X(u, w)\leq \max\{d_X(u, v), d_X(v, w)\}$.
\end{definition}

\begin{definition}[$k$-HST] \label{hst}
    Let $k >1$ be a fixed constant. A metric space $(X, d_X)$ is a {\em $k$-hierarchical well-separated tree ($k$-HST)} if there exists a bijection $\phi$ from $X$ to leaves of a rooted tree $T$ in which:
    \begin{enumerate}
        \item Each node $v\in T$ is associated with a label $\Gamma(v)$ such that $\Gamma(v)=0$ if $v$ is a leaf, and $\Gamma(v)\geq k \cdot \Gamma(u)$ \textnormal{($\Gamma(v) > 0$)} if $v$ is an internal node and $u$ is a child of $v$.
        \item $d_X(x, y) = \Gamma(\texttt{LCA}(\phi(x), \phi(y)))$, where $\texttt{LCA}(u, v)$ is the lowest common ancestor of $u, v$.
    \end{enumerate}
\end{definition}

Lastly, we defer some proofs of lemmas/claims (marked by $(\star)$) to~\Cref{sec:omitted_proofs}.

\section{Sparsity and Lightness Lower Bounds for $(1+\varepsilon)$-Spanners.~}
\label{sec:lower_bound}

This section is devoted to the proof of \Cref{thm:main}. We state it again below:
\LowerBoundMain* 

We begin by constructing an ultrametric space $(X,d_X)$ for which we later prove that any $(1+\eps)$-spanner must satisfy the size and lightness lower bound. 

Without loss of generality, assume that $n$ is a power of $2^d$. More specifically, if $n$ is not a power of $2^d$, let $\tilde{n} = (2^d)^\ell$ be the largest integer power of $2^d$ before $n$, with $\tilde{n} \ge n/2^d$, and apply the argument to $\tilde{n}$ instead of $n$. The size lower bound will decrease by at most an $O(1)$ factor $2^d$ and the lightness lower bound will remain the same.

We define the following 2-HST $(X, d_X)$.
Let $T$ be a rooted perfect $2^d$-ary tree with $n$ leaf nodes: all leaf nodes are at the same depth $\log_{2^d} n = \frac{\log_2 n}{d}$ and each internal node has $2^d$ children. 
The {\em height} of a node is its distance to the closest leaf, so the leaf nodes have height 0, their parents have height 1, and so on, until reaching the root, at height $\frac{\log_2 n}{d}$.
For each node $v$ in the tree we assign a label, denoted by $\Gamma(v)$.
Each leaf node is assigned label 0, and each node of height $i = 1,\ldots,\frac{\log_2 n}{d}$ is assigned label $2^i$.
For any $x, y \in T$, we denote the {\em lowest common ancestor} of $x$ and $y$ by $\texttt{\texttt{LCA}}(x, y)$. 
Let $\psi$ be any arbitrary bijection mapping $X$ to the leaf nodes of $T$. For any $u, v \in X$, we define $d_X(u, v)$ to be $\Gamma(\texttt{\texttt{LCA}}(\psi(u),\psi(v)))$.

We start with the following basic observation in \Cref{clm:bdd-dd}. Although to the best of our knowledge, this precise statement does not appear in previous work, weaker (though similar) versions of this statement can be found in~\cite{CG06,FL22}. 
\begin{claim}
\label{clm:bdd-dd}
$(X, d_X)$ is an ultrametric space of doubling dimension $d$. 
\end{claim}
\begin{proof}
First, we argue that $(X, d_X)$ is an ultrametric space. By the definition of $d_X$, we have $d_X(u,v) = 0$ iff $u = v$. Further, $d_X(u,v) = d_X(v,u)$ holds for any $u,v \in X$ since $\texttt{LCA}(\psi(u),\psi(v))=\texttt{LCA}(\psi(v),\psi(u))$. 
It remains to show that 
$d_X(u,v) \leq \max\{d_X(u,w), d_X(w,v)\}$, for any $u,v,w\in X$.
Let $T_u$ ($T_v$) denote the subtree rooted at the child node of $\texttt{LCA}(\psi(v),\psi(u))$ that contains $\psi(u)$ ($\psi(v)$). Note that for any $w \in X$, $\psi(w)$ cannot belong to both $T_u$ and $T_v$, so either $\texttt{LCA}(\psi(v),\psi(w))$ is not contained in $T_v$, or $\texttt{LCA}(\psi(u),\psi(w))$ is not contained in $T_u$.
Thus, $d_X(u,v) \leq \max\{d_X(u,w), d_X(w,v)\}$ holds.

Second, we prove that $(X, d_X)$ has doubling dimension $d$. 

We first show that the doubling dimension is no smaller than $d$. Let $B(u, r)$ be a ball centered at $u \in X$ with radius $r \in \mathbb{R}^+$. Let $\psi(f_u)$ be the parent of $\psi(u)$ in $T$. Consider $r=2$. By the construction of $T$, every point $v\in X$ corresponding to a child of $\psi(f_u)$ satisfies $d_X(u,v) = \Gamma(\psi(u),\psi(v))$. Hence, $B(u,2)$ contains all points corresponding to the children of $\psi(f_u)$, and contains at least $2^d$ distinct points. On the other hand, for any $v\in X$, the ball $B(v,1)$ contains only the single point $v$. Hence, covering $B(u,2)$ requires at least $2^d$ balls of radius $1$. This implies that the doubling dimension is no smaller than $d$.

Then, we show that the doubling dimension of $(X, d_X)$ is no bigger than $d$. When $r < 2$, $B(u, r)$ contains only one point $u$ and is covered by $B(u, \frac{r}{2})$. When $2^i \le r < 2^{i + 1}$ for some integer $i \geq 1$, let $\omega$ be the ancestor of $\psi(u)$ of height $i$ in $T$ and let $T_{\omega}$ be the subtree rooted at $\omega$. By the construction of $T$, all pairwise distances between the points corresponding to the leaf nodes of $T_{\omega}$ are no larger than the label $2^i$ of $\omega$. On the other hand, the distance between $u$ and any point corresponding to a leaf node outside $T_{\omega}$ is no smaller than the label $2^{i+1}$ of the parent node of $\omega$, as their lowest common ancestor is outside $T_{\omega}$. Thus, $B(u, r)$ consists of the $(2^d)^i = 2^{d\cdot i}$ points that correspond to the leaf nodes of $T_{\omega}$. 
Let $L_j$ be the set of points corresponding to the leaves of the subtree rooted at the $j$th child of $\omega$, for each $j \in [2^d]$, and let $v_j$ be an arbitrary point in $L_j$. Since all children of $\omega$ have label $2^{i-1}$, all pairwise distances between the points in $L_j$ are no larger than $2^{i-1}$. So $B(v_j, \frac{r}{2}) \supseteq B(v_j,2^{i-1})$ covers all $(2^d)^{i-1} = 2^{d\cdot (i-1)}$ points in $L_j$, for each $j \in [2^d]$, implying that the ball $B(u,r)$ can be covered by the $2^{d}$ balls 
$B(v_j, \frac{r}{2}), j \in [2^d]$.
\end{proof}

Then, we demonstrate the following property of $(X,d_X)$ in \cref{cl:include_edges} that will be crucial for proving the size and lightness lower bounds. 
\begin{claim} \label{cl:include_edges}
    Let $G= (X,E,w)$ be an arbitrary $(1+\eps)$-spanner for $(X,d_X)$. For any point pair $u,v \in X$ such that $d_X(u,v) < \frac{2}{\eps}$,  edge $(u,v)$ must be in $G$.
\end{claim}
\begin{proof}
    Consider any pair $u,v \in X$. 
    We first argue that if edge $(u,v)$ is not in $G$, then $d_G(u,v) \ge d_X(u,v) + 2$. In this case, consider a shortest path from $u$ to $v$ in $G$, and note that it traverses an intermediate point $w \in X$. By the triangle inequality, $d_G(u,v) \ge d_G(u,w) + d_G(w,v) \ge d_X(u,w) + d_X(w,v)$. 
Since the minimum pairwise distance is $2$ and  $(X,d_X)$ is an ultrametric, 
$d_G(u,v) \ge d_X(u,w) + d_X(w,v) =
\max\{d_X(u,w),d_X(w,v)\}
+ \min\{d_X(u,w),d_X(w,v)\}
 \ge d_X(u,v) + 2$.

We conclude that if $d_X(u,v) < \frac{2}{\eps}$, then edge $(u,v)$ must be in $G$, otherwise by the above assertion we would get $d_G(u,v) \ge d_X(u,v) +2 > (1+\eps)d_X(u,v)$, which is a contradiction to $G$ being a $(1+\eps)$-spanner for $(X,d_X)$.
\end{proof}

In the following subsections, we apply \cref{cl:include_edges} to prove the size and lightness lower bounds of any $(1+\eps)$-spanner $G= (X,E,w)$ for $(X,d_X)$. 

\vspace{5pt}
\subparagraph{Size lower bound.~}

Let $D_{\max}$ be the largest label among nodes in $T$ that is strictly smaller than $\frac{2}{\eps}$. Since the labels in $T$ grow geometrically by a factor of $2$, there exists a label in $\bigl[\frac{1}{\eps},\,\frac{2}{\eps}\bigr)$, hence $\frac{1}{\eps}\le D_{\max} <\frac{2}{\eps}$. Let $i_{\text{max}} = \log(D_{max}) \ge \log(\frac{1}{\eps})$ be the height of   nodes with label $D_{max}$. Each subtree of height $i_{\text{max}}$ has
$(2^d)^{i_{\text{max}}} \ge \eps^{-d}$ leaf nodes (this is where the restriction $\eps^{-d} = O(n)$ comes into play) and the pairwise distances between the corresponding points are bounded by $D_{max} < \frac{2}{\eps}$, hence \Cref{cl:include_edges}
implies that the spanner $G$ connects all these points by a clique.
Thus each point in $X$ has a degree of at least $\eps^{-d}-1$ in $G$, and so $G$ contains at least  $n (\eps^{-d}-1)/2 = \Omega(n \cdot \eps^{-d})$ edges, which proves the size lower bound. 

\vspace{5pt}
\subparagraph{Lightness lower bound.~}
We begin by giving the lower bound on the weight of $G$. Consider an arbitrary point $u \in X$ and denote the subtree rooted at the ancestor of $\psi(u)$ at level $i_{\text{max}}$ as $T_{i_{\text{max}}}(u)$. We have shown that $u$ is connected to all 
$(2^d)^{i_{\text{max}}} \ge \eps^{-d}$
points corresponding to the leaf nodes in $T_{i_{\text{max}}}(u)$. Note also that only $(2^d)^{i_{\text{max}}-1}$ points among those may
belong to the child subtree of $T_{i_{\text{max}}}(u)$ that contains $\psi(u)$, while the remaining $(2^d)^{i_{\text{max}}} - (2^d)^{i_{\text{max}}-1}  \ge \eps^{-d}/2$ points belong to other child subtrees of $T_{i_{\text{max}}}(u)$, and are therefore at distance $D_{max} = 2^{i_{\text{max}}} \ge 1/\eps$ from $u$.
Thus the total weight of edges incident on any point $u \in X$ in $G$ is at least $\eps^{-d}/2 \cdot (1/\eps) = \eps^{-(d+1)}/2$, and so the weight of $G$ is at least $n \cdot \eps^{-(d+1)}/4
= \Omega(n \cdot \eps^{-(d+1)})$.

\Cref{cl:MST} concludes the proof of the lightness lower bound for $d \ge 2$. As for $d = 1$, \Cref{cl:MST} is not enough, since we need $w(\mst_X)$ to be $O(n \cdot \log(\eps^{-1}))$ rather than $O(n \log n)$ to obtain the required lightness bound; we will handle this issue after the proof of \Cref{cl:MST}.
\begin{claim} \label{cl:MST}
$w(\mst_X) = O(n)$ for any $d \ge 2$ and $w(\mst_X) = O(n \log n)$ for $d = 1$.
\end{claim}
\begin{proof}
Let $P_X$ be the Hamiltonian path of $X$ that traverses the points of $X$ according to the induced left-to-right ordering in $T$ of the corresponding leaves. To prove the claim, we will show that $w(P_X) = O(n)$ for any $d \ge 2$ and $w(P_X) = O(n \log n)$ for $d = 1$. (It can be shown that $P_X = \mst_X$, but there is no need for that, since we only need to upper bound $w(\mst_X)$ and we have $w(P_X) \ge w(\mst_X)$.)

Denote the root of $T$ by $r$, and note that $\Gamma(r) = 2^{\frac{\log_2 n}{d}} = n^{1/d}$. Observe that the number of edges in $P$ of weight $\Gamma(r) = n^{1/d}$ is exactly one less than the number of children, i.e., $2^d - 1$. In general, note that the number of edges in $P$ of weight $\frac{\Gamma(r)}{2^i} = \frac{n^{1/d}}{2^{i}}$ is $2^{d \cdot i}(2^d - 1)$, for each $i \in [0,\frac{\log_2 n}{d}-1]$. 
It follows that 
\begin{align*}
     w(P_X) = \sum_{i = 0}^{\frac{\log_2 n}{d}-1} {2^{d \cdot i}(2^d - 1) \cdot \frac{n^{1/d}}{2^{i}}}  
     ~=~  (2^d-1) \cdot n^{1/d}  \cdot \sum_{i = 0}^{\frac{\log_2 n}{d}-1} {(2^{d-1})^i}.
\end{align*}
If $d = 1$, we get that $w(P_X) = n \cdot \log_2 n.$
\\If $d \ge 2$, then 
$\sum_{i = 0}^{\frac{\log_2 n}{d}-1} {(2^{d-1})^i}$ is a geometric sum with rate $2^{d-1} \ge 2$, hence we have
$$\sum_{i = 0}^{\frac{\log_2 n}{d}-1} {(2^{d-1})^i} ~\le~
2 \cdot 2^{(d-1) \cdot (\frac{\log_2 n}{d}-1)} ~=~ 2 \cdot 2^{\log_2 n - d - \frac{\log_2 n}{d} + 1} ~=~ 4 \cdot \frac{n}{2^d \cdot n^{1/d}},$$ and we get that
$w(P_X) = (2^d-1) \cdot n^{1/d}  \cdot \sum_{i = 0}^{\frac{\log_2 n}{d}-1} {(2^{d-1})^i} \le 4n$. 
\end{proof}
Finally, consider the case $d = 1$.
\Cref{cl:MST} gives $w(\mst_X) = O(n \log n)$, but to get the required lightness bound we need an upper bound of $w(\mst_X) = O(n \cdot \log(\eps^{-1}))$.

Without loss of generality, we may assume that $n = n'$ and that $n' = \Theta(\eps^{-1})$. In this case, \Cref{cl:MST} implies that $w(\mst_X) = O(n' \log n') = O(n' \log(\eps^{-1}))$, and the lightness of $G$ is 
$\Omega(\frac{\eps^{-2}}{\log(\eps^{-1})})$,
as required. We introduce two different reductions justifying this assumption.

First, one can simply take the same ultrametric as before on top of $n'$ points (instead of $n$) as leaves, and then add $n-n'$ points that are arbitrarily close to one of the $n'$ points, to get a metric space with $n$ points and basically the same MST weight; while this reduction provides the required lightness lower bound, it does not preserve the size lower bound.

To get a single instance for which both lower bounds apply, one can do the following. Create $n / n'$ vertex-disjoint copies of the same ultrametric $(X,d_X)$ as before, each on top of $n'$ points as leaves, denoted by $(X^{(j)},d_{X^{(j)}})$, for each $j \in [n/n']$, and place these copies ``on a line'', so that neighboring copies 
$(X^{(j)},d_{X^{(j)}})$ and $(X^{(j+1)},d_{X^{(j+1)}})$
are at distance say $2n'$ from each other. Since any two copies are sufficiently spaced apart from each other (w.r.t.\ the maximum pairwise distance in each copy, namely $n'$), any $(1+\eps)$-spanner edge between two points in the same copy must be   contained in that copy, hence we can apply the aforementioned  size and weight spanner lower bounds on each of the copies separately, and then aggregate their size and weight bounds to achieve the same size and weight lower bounds as before, namely
$\Omega(n \cdot \eps^{-1})$ and $\Omega(n \cdot \eps^{-2})$, respectively,
for the entire spanner. Since each copy has weight $O(n' \log n')$ by \Cref{cl:MST} and as neighboring copies are sufficiently close to each other, 
the MST weight of this metric is bounded by $O(n' \log n') \cdot (n / n') = O(n \cdot \log(\eps^{-1}))$, which proves the required lightness lower bound $\Omega(\frac{\eps^{-2}}{\log(\eps^{-1})})$.

\section{Bounded Degree Tree Cover} \label{sec:treecover}
\TreeCoverBoundedDegree*

\subsection{Phase 1: Unbounded Degree Tree Cover}\label{subsec:treecover_nodegbound}

\begin{lemma}
    For any metric $(X, d_X)$ with doubling dimension $d$, and any parameter $0<\varepsilon<1$, there is a $(1+\varepsilon)$-stretch tree cover $\mathcal{T}$ of size $O\left({\varepsilon}^{-d}\cdot \log(1/\varepsilon)\right)$, where the maximum degree of a vertex in each of the trees in $\mathcal{T}$ is $O(\varepsilon^{-d}\log_{1/\eps} \Delta)$.
\end{lemma}
We first give a construction of a tree cover with maximum degree $O(\eps^{-d}\log_{1/\eps}\Delta)$. Later, in~\Cref{sec:treecover_bounded_deg}, we modify this tree cover to get constant degree bound.

\begin{enumerate}
    \item \textbf{[Creating intra-level edge matchings.]} Let $\delta_p=2^p$ for $p\in\{0,\ldots, \floor{\log(1/\eps)\}}$, 
    and let $\tau_{\delta_p}$ be a $(\delta_p, \eps^{-1})$-net tree (see~\Cref{def:alpha-beta-net-tree}) of $X$ and $N_i$ be the $\delta_p\eps^{-i}$-net in $\tau_{\delta_p}$. For each $p$, construct graph $G_p$, where $V(G_p) = \{u_i \mid u\in N_i\}_{\forall i\in \mathbb{N}}$ and $E(G_p)=\emptyset$.
    Define a set of \textit{cross edges} $E_c = \{(u_i,v_i) \mid d_X(u,v) \in (\frac{\delta_p}{3\varepsilon^{i+1}}, \frac{4\cdot \delta_p}{3\varepsilon^{i+1}})\}_{\forall i}$.  
    Partition $E_c$ into $\gamma = O(\varepsilon^{-d})$ matchings $\mathcal{M} = \{M_1, \dots, M_\gamma\}$, with the property that in each $M\in \mathcal{M}$, $d_X(u_i,v_i)> \frac{14\delta_p}{\eps^i}$, for all matched $u_i$, $v_i$. Initialize $T_k^p = (V(G_p), M_k)$ for $k\in\{1,\ldots, \gamma\}$.

    \item \textbf{[Creating inter-level edges.]} For each $T=T_k^p$, construct edges between levels $i$ and $(i-1)$ in the order $i=0,1,\ldots$. Create a relabel function $\Phi_i^T: N_i\xrightarrow{} N_i$, initialized as $\Phi_i^T(u_i) = u_i$ for all $u_i\in N_i$. At each level $i$, process vertices in the order specified below: 
    
    \begin{enumerate}
        \item \textbf{[Matched vertices.]} For each matched $u_i$ in $M_k$, add $(u_i, u_{i-1})$. For $v\in B_X(u, \frac{4\delta_p}{\eps^i})$ such that $v_{i-1}$ exists, add $(u_i, v_{i-1})$ if either $v_{i-1}$ is unmatched, or $(v_{i-1}, x_{i-1})\in M_k$ and $x_{i-1}$ does not have a parent at level $i$. 
        
        \item \textbf{[Unmatched vertices.]} For each unmatched $u_i$: enumerate its child nodes in net-tree $\tau_{\delta_p}$, denoted by $C_{u_i}=\{a_{i-1}^0, a_{i-1}^1,\ldots, a_{i-1}^\eta\}$. For every $a_{i-1}\in C_{u_i}$:
        \begin{itemize}
            \item Add $(\Phi_i^T(u_i), a_{i-1})$ in the following cases: 1) $a_{i-1}$ is unmatched and does not have a parent, and 2) $(a_{i-1}, b_{i-1})\in M_k$ and both $a_{i-1}$, $b_{i-1}$ do not have parents at level $i$. 
            \item Otherwise, if $a_{i-1}$ has a parent $x_i$, or $(a_{i-1}, b_{i-1})\in M_k$ and $b_{i-1}$ has a parent $x_i$, merge $x_i$ with $\Phi_i^T(u_i)$. Set $\Phi_i^T(x_i)=\Phi_i^T(u_i)$ if $u_{i+1}$ exists, otherwise set $\Phi_i^T(u_i)=x_i$. (Note that $a_{i-1}$ and $b_{i-1}$ will not have parents at the same time.)
        \end{itemize}
        
    \end{enumerate}

    \item \textbf{[Contraction.]} Contract all vertices sharing label $u$. 
    $\mathcal{T} =  \{T_k^p\}_{p,k}$
    forms a tree cover.
\end{enumerate}

Throughout the construction (even after relabeling), the weight of any edge, $w_T(u_i, v_j) = d_X(u, v)$. In Step 2, we slightly abuse the notation, setting $\Phi_i^T(u_i)= x_i$, instead of $\Phi_i^T(\Phi_i^T(u_i))= x_i$. This is fine, since a node undergoes relabeling at most once. This finishes the construction of the tree cover.

Unless stated otherwise, we refer to the intermediate trees just prior to Step 3 as the tree cover $\mathcal{T}$ and prove our bounds on them. This works since the contraction of only 0-weight edges in Step 3 does not affect the stretch.

We first prove that each $T$ in $\mathcal{T}$, after contraction, is a tree (see~\Cref{cor: tree_cover_is_tree}). Then, we bound the number of trees in $\mathcal{T}$ (see~\Cref{lem:number_of_trees}). Finally, we use~\Cref{clm:edge_weight_bounds} and~\Cref{clm:dist_leaf_to_ancestor} to bound the final stretch in~\Cref{lem:final_stretch_bound_ph1} (proof deferred to~\Cref{sec:omitted_proofs}).

We first require the following combinatorial lemma of~$\cite{FL22}$, used in Step 1 of our construction. 

\begin{lemma}[Lemma 6 of~\cite{FL22}]\label{lem:HL_comb} 
    Consider a graph $G = (V, E_b \cup E_r)$ with disjoint sets of blue edges $E_b$ and red edges $E_r$, such that the maximum blue degree is $\delta_b\geq 1$, and the maximum red degree is $\delta_r\geq 1$. Then there is a set of at most $\gamma = O(\delta_b\delta_r)$ matchings $\mathcal{M} = \{M_1, \ldots, M_\gamma\}$ of $G$ such that: 
    \begin{enumerate}
        \item $E_b\subseteq \cup_{i=1}^\gamma M_i$, and
        \item for every matching $M\in \mathcal{M}$, there is no red edge whose both endpoints are matched by $M$. 
    \end{enumerate}
\end{lemma}

\begin{remark}
Note that we apply this lemma in Step 1 of the construction. More specifically, $E_b = E_c = \{(u_i,v_i) \mid d_X(u,v) \in (\frac{\delta_p}{3\varepsilon^{i+1}}, \frac{4 \delta_p}{3\varepsilon^{i+1}})\}_{\forall i}$, and $E_r = {\{(u_i,v_i) \mid d_X(u,v) \leq  \frac{14\delta_p}{\eps^i}}\}$. 
In each $M\in \mathcal{M}$, $d_X(u_i,v_i)> \frac{14\delta_p}{\eps^i}$ holds for all matched $u_i$, $v_i$.
\end{remark}

We begin by proving that each $T_k^p\in \mathcal{T}$ is a tree. We rely on the following inductive invariant: after processing the $i^{th}$ level, every level-$i$ subtree has a diameter at most $O(\eps)$ factor of the newly added level-$i$ cross edge. This large separation ensures that the forest remains a forest. 

\begin{restatable}{lemma}{IntermediateTree}
     Each $T_k^p\in \mathcal{T}$ is a tree.
\end{restatable}
\begin{proof}

Let $T=T_k^p$. We will prove the following two properties. First, each vertex at level-$i$ is connected to at most one vertex at level-$(i+1)$. Second, for every level-$i$ cross edge, at most one of the end points is connected to a level-$(i+1)$ vertex. These two arguments are sufficient condition for $T$ to be acyclic since all cross edges are mutually vertex-disjoint.

We begin by proving that these two properties hold in Step 2a. For any vertex $v_{i-1} \in B_X(u, \frac{4\delta_p}{\eps^i})$ and $u_i$ is a matched vertex of $T$, $v_i$ cannot also be a matched vertex of $T$ since $d_X(u,v) > \frac{14\delta_p}{\eps^i}$ should hold by \cref{lem:HL_comb}. Hence, any vertex $v_{i-1}$ is connected to at most one vertex at level-$(i+1)$. For any matched cross edge $(u_{i-1}, v_{i-1}) \in T$, we have $d_X(u,v) \in (\frac{\delta_p}{3\varepsilon^{i}}, \frac{4 \delta_p}{3\varepsilon^{i}})$, in which case $u_i$ and $v_i$ cannot both be matched vertices in level-$i$ since this requires $d_X(u,v) >  \frac{14\delta_p}{\eps^i}$. Thus, $(u_i, u_{i-1})$ and $(v_i, v_{i-1})$ will not be added at the same time, and $(u_{i-1}, v_{i-1})$ cannot have parents at both endpoints. Next, by construction (Step 2b), the two properties hold. This is because an edge between some $a_{i+1}$ and $b_i$ is created only when either $b_i$ is unmatched and does not have a parent, or some $(b_i, c_i)$ exists and neither of $b_i$ or $c_i$ have a parent. At each level of the net tree $\tau_{\delta_p}$, we go over all vertices.

Let $rt$ be the vertex at the highest level of $T$. For any two vertices $u, v\in V(T)$, we know that there is a path from $u$ to $\textit{rt}$, and from $v$ to $\mathit{rt}$ in $T$. This holds because each vertex is either connected to a vertex one level above, or to another vertex at the same level via a cross edge (since each such vertex is enumerated in Step 2b). Hence, there is a walk between $u$ and $v$ via $\textit{rt}$ in $T$.
This implies that $T$ is connected. 
\end{proof}

\begin{restatable}{corollary}{EachTisTree}{$(\star)$}
\label{cor: tree_cover_is_tree}
    After the contraction procedure, each $T_k^p$ in $\mathcal{T}$ remains a tree.
\end{restatable}

\begin{lemma}\label{lem:number_of_trees}
    The number of trees in $\mathcal{T}$ is $O(\eps^{-d}\log(1/\eps))$.
\end{lemma}
\begin{proof}
    For each $G_p$ ($p\in \{0,1\ldots,\lfloor\log(1/\eps)\rfloor\}$, we take $V = V(G_p)$, $E_b=E_c$. We define a set of dummy edges in $G_p$: $E_r = \{(u_i, v_i)\mid d_X(u, v)\leq \frac{14\delta_p}{\eps^i}\}_{\forall i}$. By packing lemma, the maximum degree in $E_b$ at each level $i$ is $O((\delta_p\eps^{-(i+1)}/\delta_p\eps^{-i})^d) = O(\eps^{-d})$.
    
    Similarly, the maximum degree in $E_r$ is $O(14\delta_p\eps^{-i}/\delta_p\eps^{-i})^d = c^d$ for some constant $c$. We can now invoke~\Cref{lem:HL_comb} on $G= (V, E_b\cup E_r)$ with $\delta_b = O(\eps^{-d})$ and $\delta_r = c^d$ to get $O((c\eps^{-1})^d) = O(\eps^{-d})$ matchings (each corresponding to a tree) with the required properties. Hence, the total number of trees in $\mathcal{T}$ is $O(\gamma\cdot \log(1/\eps)) = O(\eps^{-d}\log(1/\eps))$.
\end{proof}

\begin{restatable}{claim}{EdgeWeightBounds}{$(\star)$}
\label{clm:edge_weight_bounds}
    Let $T=T_k^p$. Let $(a_i, b_i)$ be a cross edge and $(c_i, d_{i-1})$ be an inter-level edge in $T$. For $\eps<1/156$, we have
    \[\frac{\delta_p}{4\eps^{i+1}} \leq d_T(a_i,b_i) \leq \frac{3\delta_p}{2\eps^{i+1}}, \quad \text{and}\quad d_T(c_i, d_{i-1})\leq \frac{27}{2}\frac{\delta_p}{\eps^i}.\]

\end{restatable}

\begin{restatable}{claim}{DistLeafToAncestor}{$(\star)$}
\label{clm:dist_leaf_to_ancestor}
    For any $T_k^p\in \mathcal{T}$, the distance from a leaf node $u_0$ to any of its ancestors at level $i$ is at most $2\delta_p\varepsilon^{-(i+1)}$.
\end{restatable}

\begin{restatable}{lemma}{FinalStretchBoundPhaseOne}{$(\star)$}
\label{lem:final_stretch_bound_ph1}
    For every $u,v\in X$, there exists a tree $T$ in the tree cover $\mathcal{T}$ such that:
    $d_T(u, v) \leq (1+\eps) d_X(u, v).$
\end{restatable}

\subsection{Phase 2: Bounding the Degree}\label{sec:treecover_bounded_deg}

The algorithm in~\Cref{subsec:treecover_nodegbound} builds the tree cover $\mathcal{T}$ of size $O(\varepsilon^{-d}\log(1/\eps))$. In this section, we focus on adjusting edges of these trees to achieve constant degree bound.
We state the main algorithm and refer to~\Cref{sec:treecover_bdd_deg_appendix} for the detailed analysis of this section. 

Let $T$ denote a tree in the tree cover $\mathcal{T}$ constructed before Step 3 (contraction).
Let $l(x)$ be the highest level of $x$ in the net tree $\tau_{\delta_p}$. We assign a direction to each edge of $T$ based on the levels of its endpoints in $\tau_{\delta_p}$. An edge $e = (a_i \rightarrow b_j)$ is going out from $a_i$ to $b_j$ if $l(a) \leq l(b)$. If $l(a) = l(b)$, we assign the direction arbitrarily. Similar to Phase 1, we categorize the set of directed edges into two categories:
\begin{enumerate}
    \item An edge $(x_i\rightarrow p_{i+1})$ is an \EMPH{inter-level edge at level $i$} if $x \in N_i \setminus N_{i+1}$ and $p \in N_{i+1}$.
    Let $D_i$ be the set of inter-level edges at level $i$. We denote the set of points $p_{i+1}$ incident to an inter-level edge $(x_i\rightarrow p_{i+1})$ at level $i$ by $D_i^-(x)$, and the set of points $w_i$ incident to an inter-level edge $(w_i\rightarrow x_{i+1})$ at level $i$ by $D_i^+(x)$. 
    \item An edge $(x_i\rightarrow y_i)$ is a \EMPH{cross edge at level $i$} if $x, y\in N_i$. Let $C_i$ be the set of all cross edges at level $i$. Given a point $x_i$, we denote the set of all points $y_i$ having a cross edge $(y_i \rightarrow x_i)$ by $C_i^+(x)$, and the set of all points $z_i$ having a cross edge $(x_i \rightarrow z_i)$ by $C_i^-(x)$. 
\end{enumerate}

By our Phase 1 construction, we observe the following bounds:
\begin{restatable}{observation}{DegBounds}{}
\label{obs:originaldeg}    
Given a fixed level $i$, for every vertex $x \in N_i$, $|D_i^+(x)| = O(\varepsilon^{-d})$, $|D_i^-(x)| \leq 1$, $|C_i^+(x)| \leq 1$, $|C_i^-(x)| \leq 1$. Furthermore, for every point $x$, we have $\sum_i |D_i^-(x)| \leq 1$, and $\sum_i |C_i^-(x)| \leq 2$.
\end{restatable}
We use the degree reduction technique of~\cite{CGMZ05} to reduce degree associated with both cross edges and inter-level edges.
The input and output of this technique are specified in~\Cref{rm:redirect_alg_input} and~\Cref{lmm:redirecting}, respectively.

\begin{restatable}{remark}{RedirectAlgInput}{}
\label{rm:redirect_alg_input}
    The input tree $T$ must be a tree in the tree cover, and we consider a subset $E$ of edges in $T$. Let $E_i$ denote the edges in $E$ at level $i$, and $E_i^{+}(x)$ (resp., $E_i^{-}(x)$) denote the set of points connected to $x$ by incoming (resp., outgoing) edges in $E_i$. The subset $E$ must satisfy the following degree constraints for every point $x$: 1) $O(1)$ in-coming edges at every level, 2) at most one out-going edge at every level, and 3) totally $O(1)$ out-going edges over all levels. Formally, $|E_i^+(x)| \leq d_{in}$, $|E_i^-(x)| \leq 1$ and $\sum_i|E_i^-(x)| \leq d_{out}$, where $d_{in} > 0$ and $d_{out} > 0$ are constants.
\end{restatable}

\begin{restatable}{lemma}{RedirectingLem}{}
\label{lmm:redirecting}
    Given a tree cover $\mathcal{T}$, let $T$ be a tree in $\mathcal{T}$ that satisfies~\Cref{rm:redirect_alg_input}. There exists an algorithm redirecting edges in $E$ of $T$ to obtain a tree $T'$ with a new edge set $E'$ such that:
    \begin{itemize}
        \item (Bounded degree.) Every point has at most $2+2d_{in} +d_{out}$  edges in $E'$.
        \item (Small stretch.) For every edge $(y, x)$ in $E$ of $T$, there exist a path $\pi$ in $T'$ from $y$ to $x$ such that $\pi$ contains edges in $E'$, and $d_{T'}(y, x) \leq (1 + O(\varepsilon)) \cdot d_X(y, x)$. 
    \end{itemize}
\end{restatable}
\vspace{.5cm}
\textbf{\underline{Algorithm.}} We use two operations: redirect cross edges and redirect inter-level edges. 

\begin{itemize}
    \item\textnormal{\textbf{[Step 1. Redirecting cross edges.]}} We apply~\Cref{lmm:redirecting} to re-direct cross edges, where $E = \bigcup_i C_i$.
    $E_i^-$ and $E_i^+$ in~\Cref{rm:redirect_alg_input} will respectively be $C_i^-$ and $C_i^+$, where $|C_i^+(x)| \leq 1, |C_i^-(x)| \leq 1, \sum_t|C_t^-(x)| \leq 2$ for every point $x$ and every level $i$ (by~\Cref{obs:originaldeg}).
    \item \textnormal{\textbf{[Step 2. Redirecting inter-level edges.]}}
    \begin{itemize}
    \item\textnormal{\textbf{Step 2.1.}} We visit levels from the bottom to top. Let $Q_i$ be the set of new edges at level $i$ created by this step. Initialize $Q_i = \emptyset$. Let $Q_i^-(p)$ be the set of points $q$ having an edge $(p \rightarrow q)$ in $Q_i$. Let $Q_i^+(p)$ be the set of points $q$ having an edge $(q \rightarrow p)$ in $Q_i$.

    Given a fixed level $i$, for every vertex $x_{i+1} \in N_{i+1}$, recall that the set of in-coming inter-level edges at level $i$ of $x$ is $D_i^+(x) = \{w_i^{(1)}, w_i^{(2)}, \ldots \}$, we know $|D_i^+(x)| \leq O(\varepsilon^{-d})$.
    We sort $D_i^+$ by increasing order of distances to $x_{i+1}$.
    Consider $w^{(j)}_i \in D_i^+(x)$:
    \begin{itemize}
        \item If $j \leq 2$: keep $(w^{(j)}_i \rightarrow x)$. 
        \item If $j>2$: let $j^* = \floor{(j-1)/2}$; remove $(w^{(j)}_i \rightarrow  x_{i+1})$ and create edge $(w^{(j^*)} _i\rightarrow  w^{(j)}_i)$ with weight $d_X(w^{(j^*)}_i , w^{(j)}_i)$. We add $(w^{(j^*)}_i \rightarrow  w^{(j)}_i)$ to $Q_i$.
    \end{itemize}
         \item\textnormal{\textbf{Step 2.2.}} We apply~\Cref{lmm:redirecting} to re-direct inter-level edges remaining after Step 2.1, where $E = \bigcup_{i} D_i$.         
         $E_i^-$ and $E_i^+$ in~\Cref{rm:redirect_alg_input} will be $D_i^-$ and $D_i^+$ respectively, where
        $|D_k^+(p)| \leq 2, |D_k^-(p)|\leq 1, \sum_t|D_t^-(p)| \leq 1$ for every point $p$ and level $k$. These bounds hold  by~\Cref{obs:originaldeg} and by Step 2.1 (when $j \leq 2$).
        \item\textnormal{\textbf{Step 2.3.}} We apply~\Cref{lmm:redirecting} to re-direct $E = \bigcup_i Q_i$ created by Step 2.1.

         $E_i^-$ and $E_i^+$ in~\Cref{rm:redirect_alg_input} will be respectively $Q_i^-$ and $Q_i^+$. 
    \end{itemize}
\end{itemize}

This finishes the construction of the bounded degree tree cover.

\bibliography{ref, RPTALGbib}

\appendix
\section{Lower Bound for Maximum Degree of a Metric Spanner}\label{sec:lowerbound_maxdeg}
\LowerBoundMaxDegree*
We construct a metric space of doubling dimension $d$ such that the maximum degree of some point $x$ in any spanner $H$ must be at least $\Omega(\eps^{-d} \cdot \log(1/\eps))$. Our metric is a $2$-HST. The idea is, we construct a set of point $x$ in a doubling metric space $X$, such that for each of $\log{1/\eps}$ distance scales, there are $O(\eps^{-d})$ points within that distance to $x$. Any spanner $H$ must contains all edges from $x$ to all points in each distance scale.

We first build the tree $T$ with label function $\Gamma$, which have $1 + \log{1/\eps} \cdot \eps^{-d} / 4^d$ leaves. Let the set of leaves in $T$ be $x, v_{i, j}$ with $i \in [1, \log{1/\eps}]$ and $j \in [1, \eps^{-d} / 4^d]$. For simplicity, we assume that $\eps^{-1}$ is an integer divided by $4$. We create $\log{1/\eps}$ nodes $y_1, y_2, \ldots y_{\log{1/\eps}}$ such that $\Gamma(y_k) = \frac{2^{k - 1}}{\eps}$. The node $y_{\log{1/\eps}}$ is the root of $T$. We add the path $x, y_1, y_2, \ldots y_{\log{1/\eps}}$ to $T$. Then, we construct $\log{1/\eps}$ subtrees of $T$, denoted by $T_1, T_2, \ldots T_{\log{1/\eps}}$, such that:
\begin{itemize}
    \item For each $i$, the leaves of $T_i$ are $v_{i, j}$ for $j \in [1, \eps^{-d} / 4^d]$
    \item Each $T_i$ is a perfect $2^d$-ary tree of height $\log{1/4\eps}$. 
    \item For each node $u$ at height $k > 0$ of $T_i$, $\Gamma(u) = 2^{k + i}$.
\end{itemize}
Finally, we connect the root of $T_i$ to $y_i$ for each $i$. Let $(X, d_X)$ be the ultrametric corresponding to $T$. See \Cref{fig:bdd_deg_HST}.

\begin{figure}[tph!]
\centerline{\includegraphics[height = 6cm]{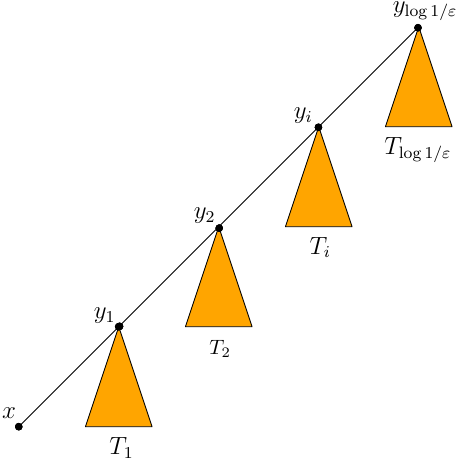}}
    \caption{Construction of $T$.}
    \label{fig:bdd_deg_HST}
\end{figure}

Using the same proof as \Cref{clm:bdd-dd}, $(X, d_X)$ is an ultrametric space of doubling dimension bounded by $d$. Consider any two points in $X$ that is not $x$. Observe that:

\begin{observation}
    \label{obs:bdd-dist}
    For any $i_1, i_2, j_1$ and $j_2$ such that $i_1, i_2 \in [1, \log {1/\eps}]$ and $j_1, j_2 \in [1, (1/4\eps)^d]$, we have:
    \begin{enumerate}
        \item \label{it:same-tree-dist} if $i_1 = i_2 = i$ and $j_1 \neq j_2$, $d_X(v_{i, j_1}, v_{i, j_2}) \geq 2^{i + 1}$, and
        \item \label{it:dif-tree-dist}if $i_1 \neq i_2$, $d_X(v_{i_1, j_1}, v_{i_2, j_2}) = 2^{\max(i_1, i_2) - 1}/\eps$.
    \end{enumerate}
\end{observation}

We then show the main lemma.

\begin{lemma}
    For any spanner $H$ of $X$, the degree of $x$ in $H$ is at least $\log{1/\eps} \cdot \eps^{-d} / 4^d$.
\end{lemma} 

\begin{proof}
    We show that for any point $v_{i, j}$ ($1 \leq i \leq \log{1/\eps}, 1\leq j \leq \eps^{-d} / 4^d$), the edge $(x, v_{i, j})$ must be in $H$. Assume $(x, v_{i, j})$ is not in $H$. Let $P$ be the shortest path from $x$ to $v_{i, j}$ in $H$. Let $v_{i', j'}$ ($1 \leq i' \leq \log{1/\eps}, 1\leq j' \leq \eps^{-d} / 4^d$) be the point before $v_{i, j}$ in $P$. Observe that $d_X(x, v_{i, j}) = \frac{2^{i - 1}}{\eps}$. On the other hand:
    \begin{equation*}
        \begin{split}
            d_H(x, v_{i, j}) &= d_H(x, v_{i', j'}) + d_X(v_{i', j'}, v_{i, j})
        \end{split}
    \end{equation*}
    If $i = i'$, by \Cref{it:same-tree-dist} of \Cref{obs:bdd-dist}, we have $d_X(v_{i', j'}, v_{i, j}) \geq 2^{i + 1}$. Thus,
    \begin{equation*}
        \begin{split}
            d_H(x, v_{i, j}) &= d_H(x, v_{i', j'}) + d_X(v_{i', j'}, v_{i, j})\\
            &\geq d_X(x, v_{i', j'}) + 2^{i + 1} = \frac{2^{i - 1}}{\eps} + 2^{i + 1} \\
            &> (1 + \eps) \cdot \frac{2^{i - 1}}{\eps} = (1 + \eps) \cdot d_X(x, v_{i, j}),
        \end{split}
    \end{equation*}
    contradicting to the fact that $H$ is a spanner of $X$.

    If $i > i'$, by \Cref{it:dif-tree-dist} of \Cref{obs:bdd-dist}, we have $d_X(v_{i', j'}, v_{i, j}) = 2^{i - 1}/\eps$. Then,
    \begin{equation*}
        \begin{split}
            d_H(x, v_{i, j}) &= d_H(x, v_{i', j'}) + d_X(v_{i', j'}, v_{i, j})\\
            &\geq d_X(x, v_{i', j'}) + \frac{2^{i - 1}}{\eps} = \frac{2^{i' - 1}}{\eps} + \frac{2^{i - 1}}{\eps} \geq \frac{1}{\eps} + \frac{2^{i - 1}}{\eps}\\
            &> (1 + \eps) \cdot \frac{2^{i - 1}}{\eps} = (1 + \eps) \cdot d_X(x, v_{i, j}) , 
        \end{split}
    \end{equation*}
    a contradiction. The penultimate equation holds because $\eps \cdot \frac{2^{i - 1}}{\eps} < \eps \cdot \frac{2^{\log{1/\eps}}}{\eps} = 1/\eps$ when $i \leq \log{1/\eps}$.

    If $i < i'$, $d_H(x, v_{i, j}) = d_H(x, v_{i', j'}) + d_X(v_{i', j'}, v_{i, j}) > d_X(x, v_{i', j'}) \geq 2\cdot d_X(x, v_{i, j}) > (1 + \eps)d_X(x, v_{i, j})$, a contradiction.

    Therefore, there is an edge from $x$ to all $v_{i, j}$, implying the degree of $x$ is at least $\log{1/\eps} \cdot \eps^{-d} / 4^d.$
\end{proof}
\section{Bounded Degree Tree Cover - Analysis}
\label{sec:treecover_bdd_deg_appendix}

Every tree $T=T_k^p\in\mathcal{T}$ is constructed from the $(\delta_p,\varepsilon^{-1})$-net tree $\tau_{\delta_p}$. 
Recall that $N_i$ (the $\delta_p\eps^{-i}$-net) is the set of points associated with nodes at level $i$ of $\tau_{\delta_p}$. 
Recall that $N_{-1} = X$, a point can appear in many levels, and $N_i \subseteq N_{i-1}$. 
For every vertex $x_i$ and its child $w_{i-1}$ in $T$, where $i \geq 0$, we have $d_X(w, x) \leq O(\frac{\delta_p}{\varepsilon^i})$. 
For every node $x_i$ and $y_i$ where $i \geq 0$, we have $d_X(x, y) > \frac{\delta}{\varepsilon^i}$ (the packing property of $\tau_{\delta_p}$). In this section, we switch between the notations $x$ and $x_i$ (denoting the vertex $x$ at level $i$) as per convenience.

By~\Cref{clm:edge_weight_bounds}, both inter-level edge and cross edge at level $i$ have weights of $O(\frac{\delta_p}{\varepsilon^{i+1}})$.  
By the packing property, we have:

\begin{claim}\label{clm:survive}
    For any inter-level or cross edge $(a \rightarrow b)$ at level $i$, $l(a) \leq i + 1$.
\end{claim}

\begin{proof}
    This claim clearly holds if $(a_i \rightarrow b_{i+1})$ is an inter-level edge. If $(a_i\rightarrow  b_i)$ is a cross edge, we have $d_X(a, b) \leq O(\frac{\delta}{\varepsilon^{i+1}})$, thus at most one point in $\{a, b\}$ appears in $N_{i+2}$. Since $l(a) \leq l(b)$, we obtain $l(a) \leq i+1$. 
\end{proof}

Recall that in the construction in~\Cref{subsec:treecover_nodegbound}, we created: (1) cross edges at level $i$ from matching two points in $N_i$, and (2) inter-level edges between level $i-1$ and $i$ by considering a ball of radius $\Theta(\frac{\delta}{\varepsilon^i})$. Notice that a point $x\in X$ has an inter-level edge at a unique level $j$ where $x_j \in N_j \setminus N_{j+1}$, thus $\sum_i|D_i^-(x)| \leq 1$.
Furthermore, once $x_j$ has an out-going edge, by~\Cref{clm:survive}, $x_{j+2}\notin V(T)$. thus $x$ has at most two out-going cross edge at levels $j$ and $j+1$. We restate the following observation:

\DegBounds*

To obtain bounded degree for every vertex, we redirect some edges as follows: we remove $e = (s \rightarrow t)$, then create an edge $e' = (z \rightarrow s)$ where $z$ is a vertex that has a path to $t$ in $T$. We say \EMPH{$e$ is redirected to $z$}, and  $(z\rightarrow s)$ is a \EMPH{redirecting edge}. We denote $s \in R^-(z)$, and $z \in R^+(s)$. 
In~\Cref{sec:redirect_alg}, we show an algorithm to redirect edges. After that, we will apply this algorithm on cross edges directly, and on inter-level edges indirectly.
 
\subsection{Redirecting Incoming Edges}\label{sec:redirect_alg} 
This section uses the degree reduction technique of~\cite{CGMZ05}. The input $T$ of this algorithm must be a tree in the tree cover $\mathcal{T}$. We will consider a subset $E$ of edges in $T$, and this set must satisfy constant out-going degree for every vertex (see~\Cref{rm:redirect_alg_input}). Our goal is to redirect edges in $E$ to obtain constant degree (both in-coming and out-going, see~\Cref{lmm:redirecting}).
An edge 
$(a_i \rightarrow b_i) \in E$ is at level $i$ if $a_i \in N_i \setminus N_{i+1}$ and $b_i \in N_i$. In this case, we denote $a \in E_i^+(b), b \in E_i^-(a)$. Here we suppose that an edge at level $i$ has a weight of $O(\frac{\delta}{\varepsilon^{i+1}})$.

\RedirectAlgInput*

Let $I_x$ be the list of levels where $x$ has an in-coming edge $(y_i\rightarrow x_i)$, i.e., if $i \in I_x$ then $|E_i^+(x)| > 0$, if $i \not\in I_x$ then $|E_i^+(x)| = 0$. 
We sort $I_x$ in increasing order.
If $i$ is the $h^{th}$ element of $I_x$, we denote $i = I_x(h)$, and $h = I_x^{-1}(i)$.

\textbf{Algorithm.~} We visit $T$ from the bottom to the top level. Given a fixed level $i$, if $(y_i\rightarrow x_i)$ is an incoming  edge of $x$, let 
 $h = I_x^{-1}(i)$, implying $(y\rightarrow x_i)$ is the $h^{th}$  edge of $x$. 

\begin{itemize}
    \item\textnormal{[Case a.]} If $h \leq 2$ we skip this step and keep $(y_i\rightarrow x_i)$. 
    \item\textnormal{[Case b.]} If $h >2$, let $i' = I_x(h-2)$ and $z_{i'}$ be the only element in $E_{i'}^+(x)$; we look at $(z_{i'}\rightarrow x_{i'})$ which is the $(h-2)^{th}$ in-coming edge of $x$. We remove $(y_i\rightarrow x_i)$ and create an edge $(z\rightarrow y)$ with weight $d_X(z, y)$.
\end{itemize}

\begin{observation}\label{obs:new_diredge}
    $(z_{i'}\rightarrow y_i)$ is the redirecting edge obtained from $(y_i\rightarrow x_i)$.
\end{observation}
\begin{proof}
    $(z_{i'}\rightarrow x_{i'})$ is an in-coming edge of $x$ at level $i' < i$.  By induction, $z$ has a path in $E$ to $x$, thus $y$ has a path in $E$ to $x$ and the claim holds.
\end{proof}

\begin{observation}\label{obs:close_dirpoint}
    $d_X(z, x) \leq O(\varepsilon) d_X(y, x)$.
\end{observation}
\begin{proof}
    Since $(y_i \rightarrow x_i)$ is an edge at level $i$, we have $y, x \in N_{i}$. By the packing property of $\varepsilon$-net tree $T$, we have $d_X(x, y) > \frac{\delta}{\varepsilon^{i}}$. In addition, since $(z_{i'} \rightarrow x_{i'})$ is an edge at level $i' \leq i-2$,  we have $d_X(z, x) \leq O(\frac{\delta}{\varepsilon^{i-1}})$. Thus the claim holds. 
\end{proof}

\begin{claim}\label{clm:bounded_edge}
    After redirecting all edges, every point $p$ has a constant number of edges incident to $p$: $\sum_{i}|E_i^+(p)| \leq 2$, $|R^+(p)| \leq d_{out}$, and $|R^-(p)| \leq 2d_{in}$.
\end{claim}
\begin{proof}
    We keep only the first and the second in-coming edge of $p$ (if it exists), thus $\sum_{i}|E_i^+(p)| \leq 2$.
    
    Recall that the algorithm redirects an edge $(y_i \rightarrow x_i)$ to $z_{i'}$, then we obtain $(z \rightarrow y)$ is a redirecting edge by~\Cref{obs:new_diredge}. 
    In this case, $|R^-(z)|$ increases by one and $|R^+(y)|$ increases by one.

    For $|R^+(p)|$, observe that it is at most the number of out-going edges that $p$ has. Thus, $|R^+(p)| \leq d_{out}$.
    For $|R^-(p)|$, let $i$ be the first level that $p$ has an out-going  edge. By~\Cref{clm:survive}, we have $l(p) \leq i+1$. In this case, there are at most two out-going edges of $p$, which are $(p \rightarrow x_1)$ at $i$ and $(p \rightarrow x_2)$ at level-$(i+1)$. For each point $x \in \{x_1, x_2\}$, we have at most $d_{in}$ in-coming edges of $x$ that we redirect to $p$. Therefore, $|R^-(p)| \leq 2d_{in}$.
\end{proof}

\begin{claim}\label{clm:stretch_crossedge}
    For every in-coming edge $(y\rightarrow x)$ in $T$, after redirecting, we obtain a path $\pi$ from $y$ to $x$ such that the total weight of $\pi$, denote as $d_\pi(y, x)$, is at most $(1 + O(\varepsilon)) d_X(x, y)$.
\end{claim}
\begin{proof}
    Suppose that we walk from $y$ to $x$ through a sequence of points $v_1, v_2 \ldots v_k$. Recall that we redirect $(y\rightarrow x)$ to $v_1$, $(v_1\rightarrow x)$ to $v_2$, and so on. Let $v_0 = y$, then we redirect $(v_i\rightarrow x)$ to $v_{i+1}$ for $i < k$. By~\Cref{obs:close_dirpoint}, $d_X(v_{i+1}, x) \leq O(\varepsilon) d_X(v_i, x)$. Thus: 
    \begin{align*}
        d_\pi(v_i, x) &= d_\pi(v_i, v_{i+1}) + d_\pi(v_{i+1}, x) = d_X(v_i, v_{i+1}) + d_\pi(v_{i+1}, x) \\
        &\leq d_X(v_i, x) + d_X(v_{i+1}, x) + d_\pi(v_{i+1}, x) \\
        &\leq d_X(v_i, x) (1 + O(\varepsilon)) + d_\pi(v_{i+1}, x) 
    \end{align*}
    By induction, suppose that $d_\pi(v_{i+1}, x) \leq (1 + O(\varepsilon)) d_X(v_{i+1}, x)$. We obtain:
    \begin{align*}
        d_\pi(v_i, x) &\leq d_X(v_i, x) (1 + O(\varepsilon)) + (1 + O(\varepsilon)) d_X(v_{i+1}, x) \\
        &\leq d_X(v_i, x) (1 + O(\varepsilon)) 
    \end{align*}    
    Thus $d_\pi(y, x) \leq (1 + O(\varepsilon)) \cdot d_X(y, x)$ as claimed.
\end{proof}

From~\Cref{clm:bounded_edge} and~\Cref{clm:stretch_crossedge}, we obtain the following result:

\RedirectingLem*

\subsection{Bounded Degree Tree Cover Construction}\label{sec:boundeddeg_tc_alg}

We restate the algorithm shown in~\Cref{sec:treecover_bounded_deg}. We use two operations: redirect cross edges and redirect inter-level edges. 

\begin{itemize}
    \item\textnormal{\textbf{[Step 1. Redirecting cross edges.]}} We apply~\Cref{lmm:redirecting} to re-direct cross edges, where $E = \bigcup_i C_i$.
    $E_i^-$ and $E_i^+$ in~\Cref{rm:redirect_alg_input} will respectively be $C_i^-$ and $C_i^+$, where $|C_i^+(x)| \leq 1, |C_i^-(x)| \leq 1, \sum_t|C_t^-(x)| \leq 2$ for every point $x$ and every level $i$ (by~\Cref{obs:originaldeg}).
    \item \textnormal{\textbf{[Step 2. Redirecting inter-level edges.]}}
    \begin{itemize}
    \item\textnormal{\textbf{Step 2.1.}} We visit levels from the bottom to top. Let $Q_i$ be the set of new edges at level $i$ created by this step. Initialize $Q_i = \emptyset$. Let $Q_i^-(p)$ be the set of points $q$ having an edge $(p \rightarrow q)$ in $Q_i$. Let $Q_i^+(p)$ be the set of points $q$ having an edge $(q \rightarrow p)$ in $Q_i$.

    Given a fixed level $i$, for every vertex $x_{i+1} \in N_{i+1}$, recall that the set of in-coming inter-level edges at level $i$ of $x$ is $D_i^+(x) = \{w_i^{(1)}, w_i^{(2)}, \ldots \}$, we know $|D_i^+(x)| \leq O(\varepsilon^{-d})$.
    We sort $D_i^+$ by increasing order of distances to $x_{i+1}$.
    Consider $w^{(j)}_i \in D_i^+(x)$:
    \begin{itemize}
        \item If $j \leq 2$: keep $(w^{(j)}_i \rightarrow x)$. 
        \item If $j>2$: let $j^* = \floor{(j-1)/2}$; remove $(w^{(j)}_i \rightarrow  x_{i+1})$ and create edge $(w^{(j^*)} _i\rightarrow  w^{(j)}_i)$ with weight $d_X(w^{(j^*)}_i , w^{(j)}_i)$. We add $(w^{(j^*)}_i \rightarrow  w^{(j)}_i)$ to $Q_i$.
    \end{itemize}
         \item\textnormal{\textbf{Step 2.2.}} We apply~\Cref{lmm:redirecting} to re-direct inter-level edges remaining after Step 2.1, where $E = \bigcup_{i} D_i$.         
         $E_i^-$ and $E_i^+$ in~\Cref{rm:redirect_alg_input} will be $D_i^-$ and $D_i^+$ respectively, where
        $|D_k^+(p)| \leq 2, |D_k^-(p)|\leq 1, \sum_t|D_t^-(p)| \leq 1$ for every point $p$ and level $k$. These bounds hold  by~\Cref{obs:originaldeg} and by Step 2.1 (when $j \leq 2$).
        \item\textnormal{\textbf{Step 2.3.}} We apply~\Cref{lmm:redirecting} to re-direct $E = \bigcup_i Q_i$ created by Step 2.1.

         $E_i^-$ and $E_i^+$ in~\Cref{rm:redirect_alg_input} will be respectively $Q_i^-$ and $Q_i^+$.          
    \end{itemize}
\end{itemize}

\begin{restatable}{observation}{DirAlgObs}{}
\label{obs:diralg_s2.1} 
At the end of Step 2.1, 
    $(w^{(j)}_i \rightarrow w^{(\lfloor{j-1)/2}\rfloor)}_i), (w^{(2j+1)}_i \rightarrow w^{(j)}_i)$ and $(w^{(2j+2)}_i \rightarrow w^{(j)}_i)$ are cross edges at level $i$. Furthermore, $|Q_k^+(p)| \leq 2$, $|Q_k^-(p)| \leq 1$, and $\sum_t|Q_t^-(p)| \leq 1$ for every point $p$ and level $k$.
\end{restatable}

\begin{proof}
    Recall that any child $w_i$ of $x_{i+1}$ has $d_X(w, x) \leq O(\frac{\delta}{\varepsilon^{i+1}})$. For any pair of children $w_{i}$ and $w'_{i}$ of $x_i$, since $l(w) = l(w') = i$, the edge between $w_i$ and $w'_i$ is a cross edge at level $i$. 
    $|Q_k^-(p)| \leq 1$ and $|Q_k^+(p)| \leq 2$ clearly holds by Step 2.1 (when $j > 2$). Since there is a unique level $i$ such that $w_i \in N_i\setminus N_{i+1}$, we obtain $\sum_t|Q_t^-(w)| \leq 1$.
\end{proof}

From Step 1, we obtain bounded cross-edge degree. From Step 2, we obtain bounded inter-level edge degree. 

\begin{observation}
    After redirecting, every point has at most a constant number of edges.
\end{observation}

Now we focus on proving the stretch.
Let $T_c, T_{1}$ be the trees obtained by Steps 1, 2.1 respectively. Let $T'$ be the final tree after we apply Steps 2.2 and 2.3.
Now we show the stretch in $T'$ of edges in $T$.

\begin{claim}\label{claim:stretch_crossedge}
    For every cross edge $(y , x)$ of $T$, $d_{T'}(y, x) \leq (1 + O(\varepsilon))\cdot d_X(y, x)$
\end{claim}
\begin{proof}
    In Step 1 we apply the algorithm in~\Cref{lmm:redirecting} directly, thus for every cross edge $(y , x)$ of $T$, there is a path in $T_c$ such that  $d_{T_c}(y, x) \leq (1 + O(\varepsilon))d_X(y, x)$. Observe that this path is preserved after Step 2 (since we apply the redirecting algorithm with $E_2 = \bigcup_i D_i$ only), thus:
    \begin{align}\label{eq:cross_stretch}
         d_{T'}(y, x) = d_{T_c}(y, x) \leq (1 + O(\varepsilon))\cdot d_X(y, x),
    \end{align}  
    as claimed.
\end{proof}

\begin{claim}\label{claim:stretch_interleveledge}
    For any inter-level edge $(w , x)$ of $T$, $d_{T'}(w, x) \leq O(d)\log{\frac{1}{\varepsilon}} \cdot d_X(w, x)$
\end{claim}
\begin{proof}
    In Step 2.1, for every inter-level edge $(w^{(j)} , x)$ of $T_1$ (also $T$), we have $d_{T_1}(w^{(j)}, x) = d_X(w^{(j)}, x)$ for $j \leq 2$, and $d_{T_1}(w^{(j)} , x) = d_{T_1}(w^{(\lfloor(j-1)/2\rfloor)}, x) + d_X(w^{(\lfloor(j-1)/2\rfloor)}, w^{(j)})$.
    By induction, suppose that $d_{T_1}(w^{(\lfloor(j-1)/2\rfloor)}, x) \leq (1 + 2\log\lfloor\frac{j-1}{2}\rfloor)\cdot d_X(w^{(\lfloor(j-1)/2\rfloor)}, x)$ for $j \geq 3$.
    Since $d_X(w^{(\lfloor(j-1)/2\rfloor)}, x) \leq d_X(w^{(j)}, x)$, we have:
    \begin{align*}
        d_{T_1}(w^{(j)} , x) &= d_{T_1}(w^{(\floor{(j-1)/2})}, x) + d_X(w^{(\floor{(j-1)/2})}, w^{(j)}) \\ 
        &\leq (1 + 2\log\lfloor\frac{j-1}{2}\rfloor)\cdot  d_X(w^{(\floor{(j-1)/2})}, x) + d_X(w^{(\floor{(j-1)/2})}, x) + d_X(w^{(j)}, x) \\
        &= (3 + 2\log\lfloor\frac{j-1}{2}\rfloor\rceil)\cdot  d_X(w^{(j)}, x)\\
        &\leq 
        (1 + 2\log{j})\cdot d_X(w^{(j)}, x)
    \end{align*}
    The vertex $x_{i+1}$ has at most $O(\varepsilon^{-d})$ children, we obtain $d_{T_1}(w^{(j)}, x) \leq  O(d)\log{\frac{1}{\varepsilon}} \cdot d_X(w^{(j)}, x)$.

    In Step 2.2, we apply the algorithm in~\Cref{lmm:redirecting} directly on inter-level edges remaining after Step 2.1. Thus, for every inter-level edge $(y , x)$ of $T_1$, there is a path in $T'$ such that: $d_{T'}(y, x) = (1 + O(\varepsilon))\cdot d_X(y, x)$.

    In Step 2.3, we apply the algorithm in~\Cref{lmm:redirecting}  on cross edges $Q_i$ created in Step 2.1. Thus for every cross edge $(y , x)$ of $T_1$, there is a path in $T'$ such that: $d_{T'}(y, x) = (1 + O(\varepsilon))\cdot d_X(y, x)$.

    From Steps 2.1, 2.2 and 2.3, for an inter-level edge $(w , x)$ of $T$, let $\pi$ be the path from $w$ to $x$ in $T_1$. 
    Every edge $(a, b) \in \pi$ is either a cross edge or an inter-level edge in $T_1$. In any case, the path from $a$ to $b$ in $T'$ has $d_{T'}(a, b) \leq (1 + O(\varepsilon)) \cdot d_X(a, b)$.
    Then: 
    \begin{equation}\label{eq:interlevel_stretch}
    \begin{aligned}
        d_{T'}(w, x) &\leq \sum_{(a, b) \in \pi} d_{T'}(a, b) \leq \sum_{(a, b) \in \pi} (1 + O(\varepsilon))\cdot d_{X}(a, b)= (1 + O(\varepsilon)) \sum_{(a, b) \in \pi} d_{X}(a, b) \\
        &= (1 + O(\varepsilon)) \cdot d_{T_1}(w, x) \leq (1 + O(\varepsilon)) \cdot O(d)\log{\frac{1}{\varepsilon}} \cdot d_X(w, x) \\
        &= O(d)\log{\frac{1}{\varepsilon}} \cdot d_X(w, x),
    \end{aligned}
    \end{equation}   
    as claimed.
\end{proof}

\begin{lemma}
    Given a tree cover $\mathcal{T}$ constructed in~\Cref{subsec:treecover_nodegbound}, let $\mathcal{T}'$ be the tree cover obtained by applying the redirecting algorithm to all trees in $\mathcal{T}$. 
    Given a tree $T \in \mathcal{T}$, for every pair of points $u, v \in X$, if $d_T(u, v) \leq (1 + O(\varepsilon))\cdot d_X(u, v)$, then the corresponding tree $T'$ of $T$ has $d_{T'}(u, v) \leq \left(1 + O(d\varepsilon\log\frac{1}{\varepsilon})\right) \cdot d_X(u, v)$.
\end{lemma}

\begin{proof} 
    Recall that for every pair of points $u, v \in X$ where $d_X(u, v) \in (\frac{\delta}{\varepsilon^{i+1}}, \frac{2\delta}{\varepsilon^{i+1}}]$, there exist $u', v' \in N_i$ such that: 1) $d_X(u, u')$ and $d_X(v, v')$ are at most $O(\frac{\delta}{\varepsilon^i})$, and 2) there exists a tree $T \in \mathcal{T}$ such that $d_T(u, v) \leq d_T(u, u') + d_T(u', v') + d_T(v', v) \leq O(\varepsilon)d_X(u, v) + d_X(u, v) + O(\varepsilon)d_X(u, v)$. Let $\Phi_i^T(u')=\lbl{u'}$, and $\Phi_i^T(v')=\lbl{v'}$.
    Observe that $(\lbl{u'}, \lbl{v'})$ is a cross edge of $T$. Consider the corresponding tree $T'$ obtained by redirecting edges of $T$. By~\Cref{claim:stretch_crossedge}:
    $$ d_{T'}(\lbl{u'}, \lbl{v'}) \leq (1 + O(\varepsilon))\cdot d_X(\lbl{u'}, \lbl{v'}) \leq (1 + O(\varepsilon)) \cdot d_X(u, v)$$
    Consider the path $\pi_{u}$ from $u$ to $u'$ in $T$, every edge $(a, b) \in \pi_u$ is either a cross edge or an inter-level edge of $T$. 
    By~\Cref{claim:stretch_crossedge} and~\Cref{claim:stretch_interleveledge}: 
    $$d_{T'}(a, b) \leq O(d)\log{\frac{1}{\varepsilon}}\cdot d_X(a, b)$$
    We obtain:
    \begin{align*}
        d_{T'}(u, \lbl{u'}) &\leq \sum_{(a,b)\in \pi_u}
        d_{T'}(a, b)
        \leq \sum_{(a,b)\in \pi_u} O(d)\log{\frac{1}{\varepsilon}}\cdot d_X(a, b) = O(d)\log{\frac{1}{\varepsilon}} \sum_{(a,b)\in \pi_u} d_X(a, b) \\
        &= O(d)\log{\frac{1}{\varepsilon}} d_T(u, \lbl{u'}) = O(\varepsilon d \log{\frac{1}{\varepsilon}}) \cdot d_X(u, v)
    \end{align*}  

    Similarly, we have:
    \begin{align*}
        d_{T'}(v, \lbl{v'}) \leq O(d)\log{\frac{1}{\varepsilon}} \cdot d_T(v, \lbl{v'}) \leq O(\varepsilon d \log{\frac{1}{\varepsilon}}) \cdot d_X(u, v)
    \end{align*} 
    Therefore:
    \begin{align*}
    d_{T'}(u, v) &\leq d_{T'}(u, \lbl{u'}) + d_{T'}(\lbl{u'}, \lbl{v'}) + d_{T'}(\lbl{v'}, v) = \left(1 + O(\varepsilon) + 2\cdot O(d\varepsilon\log{\frac{1}{\varepsilon}})\right)\cdot d_X(u, v) \\
    &= \left(1 + O(d\varepsilon\log{\frac{1}{\varepsilon}})\right) \cdot d_X(u, v)
    \end{align*}
    This completes the proof.
\end{proof}

\section{Omitted Proofs}
\label{sec:omitted_proofs}

\begin{restatable}{lemma}{Hierarchy}{}
\label{lem:hierarchy}
    In each $T\in \mathcal{T}$, vertices having the same label form a connected component.
\end{restatable}

\begin{proof}
    Let $T=T_k^p$. For a vertex with label $u$, let $l(u)$ be the highest level at which it exists in the net tree $\tau_{\delta_p}$. We first show that the vertices $u_0, u_1,\ldots, u_{l(u)-1}$ all exist in $V(T)$. A vertex $u_i$ is only relabeled when $u_{i+1}$ does not exist. Therefore, none of these vertices ever get relabeled. For some $u_i$, where $i< l(i)$, we form the edge $(u_i, u_{i-1})$ in Step 2b of the construction, no matter whether $u_i$ is matched or unmatched. At level $l(u)$, the edge $(u_{l(u)-1}, u_{l(u)})$ is added only when $u_{l(u)}\in V(T)$. Therefore, all vertices of the same label $u$ form a connected component.
\end{proof}

\EachTisTree*
\begin{proof}
    \Cref{lem:hierarchy} implies that in each $T_k^p$, vertices with the same label form a connected subgraph. The merge operation in Step 3 of the construction ensures that we have edge contractions, which implies that the resulting graph is a minor of $T_k^p$. Therefore, each $T_k^p$ remains a tree. 
\end{proof}

\EdgeWeightBounds*
\begin{proof}
    We apply induction on the levels of $T$. Assume hypothesis holds true for cross edges and inter-level edges till levels $0,1,\ldots, i-1$. 
    
    We first derive a bound on the inter-level edges. Step 2a of the construction bounds $d_T(c_i, d_{i-1})$ by $4\delta_p\eps^{-i}$ for any matched vertex $c_i$, before considering unmatched vertices. We know that any vertex is relabeled at most once in Step 2b. This is because all potential relabeled vertices of $c_i$ are within $O(\delta_p\eps^{-i})$ distance of $c$, so at most one of these is exists at level $(i+1)$.
    
    We bound $d_T(c_i, d_{i-1})$ where $c_i$ is a matched vertex. There can be multiple unmatched vertices relabeled to $c_i$. Consider $b_i$ such that $\Phi_i^T(b_i) = c_i$. For any $w_{i-1}$ such that $(b_i, w_{i-1})$ is an edge, the relabeled edge becomes $(c_i, w_{i-1})$. By covering property, $d_T(b_i, w_{i-1})\leq \delta_p\eps^{-i}$. There exists either a path $\pi_1 = (c_i,x_{i-1},y_{i-1},b_i,w_{i-1})$ or $\pi_2 = (c_i,x_{i-1},b_i,w_{i-1})$ in $T$. 
    
    Applying induction hypothesis on level-$(i-1)$ cross edge $(x_{i-1}, y_{i-1})$, we get $\text{len}(\pi_1) \leq 4\delta_p\eps^{-i} + (3/2)\delta_p\eps^{-i} + \delta_p\eps^{-i} + \delta_p\eps^{-i} \leq (15/2)\delta_p\eps^{-i}$. Similarly, $\text{len}(\pi_2)\leq 6\delta_p\eps^{-i}$. Therefore, $d_T(c_i, w_{i-1})\leq(15/2)\delta_p\eps^{-i}$. Now, $c_i$ can be relabeled one final time to some $a_i$ if $a_{i+1}$ exists. The vertex $a_i$ is unmatched before $c_i$ is relabeled. We must have $d_T(c_i, a_i)\leq d_T(c_i, p_{i-1}) + d_T(p_{i-1}, q_{i-1}) + d_T(q_{i-1}, a_i)\leq (13/2)\delta_p\eps^{-i}$. Therefore, any edge $(c_i, w_{i-1})$ is relabeled to $(a_i, w_{i-1})$, and $d_T(a_i, w_{i-1})\leq d_T(a_i, c_i) + d_T(c_i, w_{i-1})\leq (27/2)\delta_p\eps^{-i}$. This is the bound on the inter-level edge weight.
    
    Now, we derive a bound on the cross edges. Step 1 of the construction bounds any cross edge $d_T(x_i, y_i)\in [(1/3)\delta_p\eps^{-(i+1)}, (4/3)\delta_p\eps^{-(i+1)}]$. We analyze the weight of the cross edge $(a_i, b_i)$ if its end points are relabeled. We know that both $a_i, b_i$ are relabeled at most once (say, from $a'_i, b'_i$ resp.). We already show above that $d_X(a_i, a'_i)\leq (13/2)\delta_p\eps^{-i}$ and $d_X(b_i, b'_i)\leq (13/2)\delta_p\eps^{-i}$. Therefore, by triangle inequality, 
    \begin{align*}
        d_T(a_i, b_i) &= d_X(a, b) \leq d_X(a, a') + d_X(a', b') + d_X(b', b) \\
        &\leq \frac{13\delta_p}{2\eps^i}
         + \frac{4\delta_p}{3\eps^{i+1}} + \frac{13\delta_p}{2\eps^i}\leq \frac{3\delta_p}{2\eps^{i+1}} \quad \text{for $\eps < \frac{1}{78}$.}
    \end{align*}
    Similarly, we can also bound $d_T(a_i, b_i)$ by below:
    \begin{align*}
        d_T(a_i, b_i) &= d_X(a, b) \geq d_X(a', b') - d_X(a, a') - d_X(b', b) \\
        &\geq \frac{\delta_p}{3\eps^{i+1}}
         - \frac{13\delta_p}{\eps^i}\geq \frac{\delta_p}{4\eps^{i+1}} \quad \text{for $\eps < \frac{1}{156}$.}
    \end{align*}
\end{proof}

\DistLeafToAncestor*
\begin{proof}
    We apply induction on the levels of $T=T_k^p$. The base case holds true for level $0$ nodes trivially. Assume the hypothesis holds true for levels $1,2\ldots,i-1$. Consider the path $\mathcal{P}$ from a leaf node $u_0$ to its ancestor $u'_i$ in $T_k^p$. Let $u''_{i-1}$ be the last node encountered on level $i-1$ in $\mathcal{P}$. By induction hypothesis, $d_T(u_0,u''_{i-1})\leq \frac{2\delta}{\varepsilon^i}$. The path from $u''_{i-1}$ to $u'_i$ is either a direct edge or the path: $u''_{i-1} \xrightarrow{} v'_i\xrightarrow{}u'_i$. By~\Cref{clm:edge_weight_bounds}, any edge from $u''_{i-1}$ to its ancestor at level $i$ will have weight at most $(27/2)\delta_p/\varepsilon^{i}$. 
    
    Moreover, any cross edge at level $i$ has a weight at most $\frac{3\delta}{2\varepsilon^{i+1}}$. 
    Hence, 
    \begin{align*}
        d_T(u_0, u'_i) &= d_{T}(u_0, u''_{i-1}) + d_{T}(u''_{i-1}, u'_i)\\ &\leq  \frac{2\delta}{\varepsilon^i} + \left(\frac{27\delta}{2\varepsilon^i} + \frac{3\delta}{2\varepsilon^{i+1}}\right) = \frac{31\delta}{2\varepsilon^i} + \frac{3\delta}{2\varepsilon^{i+1}} \leq \frac{2\delta}{\varepsilon^{i+1}} \quad
        \text{for $\varepsilon<\frac{1}{31}.$}
    \end{align*}
\end{proof}

\FinalStretchBoundPhaseOne*
\begin{proof}
    For some integer $i$, we have $d_X(u,v) \in \left(\frac{\delta}{2\varepsilon^{i+1}}, \frac{\delta}{\varepsilon^{i+1}}\right)$ for some $\delta = 2^p$ where $p\in \{0,1,\ldots, \lfloor\log(1/\eps)\rfloor\}$
    Let $u'$ be the ancestor of $u$, and $v'$ be the ancestor of $v$ at level $i$ of the $\delta\eps^{-1}$-net tree. We know that $d_X(u, u') \leq \frac{2\delta}{\varepsilon^i}$ and $d_X(v, v')\leq \frac{2\delta}{\varepsilon^i}$ (since the distance from a vertex $u$ at level $i$ to any of its descendants is bounded by $2\delta \eps^{-i}$).
    Since $d_X(u, v) \geq \frac{\delta}{2\varepsilon^{i+1}}$, we have that $d_X(u, u') + d_X(u', v') + d_X(v', v) \geq d_X(u, v)$ (by triangle inequality). This implies that $d_X(u', v') \geq d_X(u, v) - (d_X(u, u') + d_X(v, v')) \geq \frac{\delta}{2\varepsilon^{i+1}} - \frac{4\delta}{\varepsilon^i} \geq \frac{\delta}{3\varepsilon^{i+1}}$ for small enough $\varepsilon$. Moreover, by the triangle inequality, $d_X(u', v')$ is bounded by $d_X(u', u) + d_X(u, v) + d_X(v, v')$. This implies that $d_X(u', v') \leq \frac{\delta}{\varepsilon^{i+1}} + \frac{4\delta}{\varepsilon^i} \leq \frac{4\delta}{3\varepsilon^{i+1}}$ for small $\varepsilon$.

    For ease of notation, denote $\lbl{u}'_i=\Phi_i^T(u'_i)$ and $\lbl{v}'_i=\Phi_i^T(v'_i)$. By construction, there exists a tree $T_j^p$ (call it $T$) in the tree cover which contains the cross edge $(\lbl{u}'_i, \lbl{v}'_i)$. This is because $d_X(u', v') \in \left(\frac{\delta}{3{\varepsilon}^{j+1}}, \frac{4\delta}{3\varepsilon^{j+1}}\right)$, which means that before any relabeling, $T$ contained  $(u'_i, v'_i)$.
    
    All we need to prove is that there exist paths from $u_0$ to $\lbl{u}'_i$ and from $v_0$ to $\lbl{v}'_i$ \textit{resp.} in $T$ having lengths at most $O(\frac{\delta}{\varepsilon^i})$. We will show that $\lbl{u}'_i$ (and $\lbl{v}'_i$ \textit{resp.}) is an ancestor of $u_0$ (and $v_0$ \textit{resp.}) in $T$. It is enough to show that $u_0$ has ancestor $u'_i$ just before relabeling $u'_i$ to $\lbl{u}'_i$ (call this graph $T'$), since after relabeling, any descendant of $u'_i$ becomes a descendant of $\lbl{u}'_i$.

    We look at the last node encountered at level-$(i-1)$ in the path from $u_0$ (and $v_0$) to its ancestor at level $i$ of $T'$. Let us call it $u''_{i-1}$ (and $v''_{i-1}$). We now bound $d_{T'}(u_0, u'_i)$ and the bound for $d_T(v_0, v'_i)$ will follow. By~\Cref{clm:dist_leaf_to_ancestor}, we have $d_{T'}(u_0, u''_{i-1})\leq \frac{2\delta}{\varepsilon^i}$. Now we show that there is a direct edge from $u''_{i-1}$ to $u'_i$ in $T'$. Since $T'$ is dominating, we have $d_X(u, u'') \leq d_{T'}(u_0, u''_{i-1}) \leq \frac{2\delta}{\varepsilon^i}$. By triangle inequality, 
    \[d_X(u'', u') \leq d_X(u'', u) + d_X(u, u') \leq \frac{2\delta}{\varepsilon^i} + \frac{2\delta}{\varepsilon^i} = \frac{4\delta}{\varepsilon^i}.\] 
    By contradiction, assume that $u''_{i-1}$ has a different parent $w'_i$. The node $w'_i$ must be an end point of a level-$i$ cross edge, since we added $(w'_i, u''_{i-1})$ before considering $u'_i$. By construction, $d_X(w', u'') = d_{T'}(w'_i, u''_{i-1})\leq \frac{4\delta}{\varepsilon^i}$. By triangle inequality, $d_X(u', w') \leq d_X(u', u'') + d_X(u'', w') \leq \frac{4\delta}{\varepsilon^i} + \frac{4\delta}{\varepsilon^i} = \frac{8\delta}{\varepsilon^i}$. However, by construction, since $d_X(u', w') \leq \frac{14\delta}{\varepsilon^i}$, both $u'_i$ and $w'_i$ do not simultaneously have a cross edge in $T'$. Hence, it must be that there is an edge $(u''_{i-1}, u'_i)$ in $T'$. Now we shift our focus back on $T$ after relabeling $u'_i$ to $\lbl{u}'_i$. By~\Cref{clm:edge_weight_bounds}, $d_{T}(u''_{i-1}, \lbl{u}'_i)\leq \frac{27\delta}{2\varepsilon^i}$. By triangle inequality, this will bound 
    \[d_{T}(u_0, \lbl{u}'_i) \leq d_{T}(u_0, u''_{i-1}) + d_{T}(u''_{i-1}, \lbl{u}'_i)
        \leq \frac{2\delta}{\varepsilon^i} + \frac{27\delta}{2\varepsilon^i} = \frac{31\delta}{2\varepsilon^i}.\]
    Similarly, $d_{T}(v_0, \lbl{v}'_i) \leq \frac{31\delta}{2\varepsilon^i}$. Putting it all together, 
    \begin{align*}
        d_{T}(u, v) &\leq d_{T}(u_0, v_0) \quad\text{(distances do not increase after contracting edges)}\\
        &\leq d_{T}(u_0, \lbl{u}'_i) + d_{T}(\lbl{u}'_i, \lbl{v}'_i) + d_{T}(\lbl{v}'_i, v_0) \\
        &\leq \frac{31\delta}{2\varepsilon^i} + d_X(\lbl{u}', \lbl{v}') + \frac{31\delta}{2\varepsilon^i} \hspace{0.5in}\\
        &\leq \frac{31\delta}{\varepsilon^i} + (d_X(\lbl{u}', u) + d_X(u, v) + d_X(v, \lbl{v}')) \quad\text{(by triangle inequality)}\\
        &\leq \frac{31\delta}{\varepsilon^i} + \left(d_X(u,v)+\frac{31\delta}{\varepsilon^i}\right)
        \leq 124\varepsilon \cdot d_X(u, v) + d_X(u, v) \quad \left(d_X(u, v)\geq \frac{\delta}{2\varepsilon^{i+1}}\right)\\
        &= (1+124\varepsilon)\cdot d_X(u, v).
    \end{align*} 
    We can scale $\varepsilon$ by a constant factor to achieve the required $(1+\varepsilon)$ bound.
\end{proof}

\end{document}